\documentclass[11pt,a4paper,fleqn]{article}

\usepackage[T1]{fontenc}
\usepackage[utf8]{inputenc}
\usepackage[english]{babel}

\usepackage[margin=3cm]{geometry}

\usepackage[
    backend=biber,
    style=numeric,
    sorting=nyt,
    maxbibnames=99,
    maxcitenames=99,
    useprefix=true,
    defernumbers=true
]{biblatex}
\usepackage[protrusion=true,expansion=true]{microtype}
\usepackage{booktabs}

\usepackage{enumitem}

\usepackage{subcaption}
\usepackage{mathtools}
\usepackage{amsthm}
\usepackage{amssymb}
\usepackage{bm}

\usepackage{xspace}
\usepackage{csquotes}
\usepackage{complexity}
\usepackage{placeins}
\usepackage[noblocks]{authblk}

\usepackage{url}
\usepackage[unicode]{hyperref}
\hypersetup{
    hidelinks,
    colorlinks=false,
    linkcolor=black,
    urlcolor=black,
    citecolor=black
}

\usepackage[capitalise,noabbrev]{cleveref}
\crefname{observation}{Observation}{Observations}
\Crefname{observation}{Observation}{Observations}
\crefname{claim}{Claim}{Claims}
\Crefname{claim}{Claim}{Claims}

\theoremstyle{plain}
\newtheorem{theorem}{Theorem}
\newtheorem{corollary}[theorem]{Corollary}
\newtheorem{lemma}[theorem]{Lemma}
\newtheorem{observation}[theorem]{Observation}
\newtheorem{proposition}[theorem]{Proposition}
\theoremstyle{definition}
\newtheorem{definition}[theorem]{Definition}
\newtheorem{remark}[theorem]{Remark}

\newtheoremstyle{claimstyle}
    {\topsep}
    {\topsep}
    {\itshape}
    {0pt}
    {\bfseries}
    {\textbf{.} }
    {5pt plus 1pt minus 1pt}
    {\textbf{\thmname{#1}\thmnumber{ #2}}\thmnote{ (#3)}}
\theoremstyle{claimstyle}
\newtheorem{claim}[theorem]{Claim}
\newenvironment{claimproof}[1][\proofname{} of Claim]{
    \pushQED{\qed}
    \normalfont %
    \trivlist
    \item[\hskip\labelsep\itshape#1.]\ignorespaces
}{
    
    \popQED\endtrivlist
    
}

\newcommand{\Q}{\ensuremath{\mathbb{Q}}\xspace}
\newcommand{\F}{\ensuremath{\mathbb{F}}\xspace}
\providecommand{\R}{}
\renewcommand{\R}{\ensuremath{\mathbb{R}}\xspace}

\newcommand{\calA}{\mathcal{A}}
\newcommand{\calB}{\mathcal{B}}
\newcommand{\calL}{\mathcal{L}}

\newcommand{\ER}{\texorpdfstring{\ensuremath{\exists\R}}{ER}\xspace}
\newcommand{\problemname}[1]{\textnormal{\textsc{#1}}\xspace}
\newcommand{\ETR}{\problemname{ETR}}
\newcommand{\trainNN}{\problemname{Train-F2NN}}
\newcommand{\ETRINV}{\problemname{ETR-Inv}}
\newcommand{\ERM}{\problemname{EmpiricalRiskMinimization}}
\renewcommand{\SAT}{\problemname{SAT}}
\renewcommand{\MIP}{\problemname{MIP}}

\newcommand{\wordRAM}{\textnormal{word RAM}\xspace}
\newcommand{\realRAM}{\textnormal{real RAM}\xspace}
\newcommand{\fullyconnected}{fully connected\xspace}

\newcommand{\twolayer}{two-layer\xspace}
\newcommand{\ReLU}{\ensuremath{\mathrm{ReLU}}\xspace}
\newcommand{\CPWL}{continuous piecewise linear\xspace}

\DeclarePairedDelimiter\abs{\lvert}{\rvert}
\newcommand{\dequiv}{:\equiv}

\usepackage[dvipsnames]{xcolor}

\begin{document}

\title{Training Fully Connected Neural Networks is \texorpdfstring{\ER}{ER}-Complete\footnote{
    An extended abstract of this paper appeared in \emph{Advances in Neural Information Processing Systems 36 (NeurIPS 2023)}~\cite{Bertschinger2023_TrainNN}.
}}

\author{Daniel Bertschinger}
\affil{ETH Zürich, Switzerland\\ \texttt{daniel.bertschinger@inf.ethz.ch}}
\author{Christoph~Hertrich\thanks{Part of the work was done while Christoph Hertrich was affiliated with London School of Economics, UK, and Goethe-Universität Frankfurt, Germany. Christoph Hertrich was supported by the European Research Council (ERC) under the European Union's Horizon 2020 research and innovation programme (grant agreements ScaleOpt--757481 and ForEFront--615640).}}
\affil{Université libre de Bruxelles, Belgium, \texttt{christoph.hertrich@ulb.be}}
\author{Paul~Jungeblut}
\affil{Karlsruhe Institute of Technology, Germany\\ \texttt{paul.jungeblut@kit.edu}}
\author{Tillmann~Miltzow\thanks{Tillmann Miltzow is supported by the Netherlands Organisation for Scientific Research (NWO) under project numbers~016.Veni.192.250 and VI.Vidi.213.150.}}
\affil{Utrecht University, The Netherlands\\ \texttt{t.miltzow@uu.nl}}
\author{Simon~Weber\thanks{Simon Weber is supported by the Swiss National Science Foundation under project no.~204320.}}
\affil{ETH Zürich, Switzerland\\ \texttt{simon.weber@inf.ethz.ch}}

\date{}

\maketitle

\begin{abstract}
    We consider the problem of finding weights and biases for a \twolayer \fullyconnected neural network to fit a given set of data points as well as possible, also known as \ERM.
    Our main result is that the associated decision problem is \ER-complete, that is, polynomial-time equivalent to determining whether a multivariate polynomial with integer coefficients has any real roots.
    Furthermore, we prove that algebraic numbers of arbitrarily large degree are required as weights to be able to train some instances to optimality, even if all data points are rational.
    Our result already applies to \fullyconnected
    instances with two inputs, two outputs, and  one hidden layer of \ReLU neurons. 
    Thereby, we strengthen a result by \citeauthor{Abrahamsen2021_NeuralNetworks} [NeurIPS 2021].
    A consequence of this is that a combinatorial search algorithm like the one by \citeauthor{Arora2018_Understanding} [ICLR 2018] is impossible for networks with more than one output dimension, unless $\NP = \ER$.
\end{abstract}

\section{Introduction}
The usage of neural networks in modern computer science is ubiquitous.
They are arguably the most powerful tool at our hands in machine learning~\cite{Goodfellow2016_DeepLearning}.
One of the most fundamental algorithmic problems associated with neural networks is to train a neural network given some training data. 
For arbitrary network architectures, \citeauthor{Abrahamsen2021_NeuralNetworks}~\cite{Abrahamsen2021_NeuralNetworks} showed that the problem is \ER-complete already for \twolayer neural networks and linear activation functions.
The complexity class \ER is defined as the family of algorithmic problems that are polynomial-time equivalent to finding real roots of multivariate polynomials with integer coefficients.
Under the commonly believed assumption that~\ER is a strict superset of \NP, this implies that training a neural network is harder than \NP-complete problems.

The result of \citeauthor{Abrahamsen2021_NeuralNetworks}~\cite{Abrahamsen2021_NeuralNetworks} has one major downside, namely that the network architecture is \emph{adversarial}:
The hardness inherently relies on choosing a network architecture that is particularly difficult to train.
Their instances could be easily trained if the networks were \fullyconnected.
This stems from the fact that they use the identity function as the activation function, which reduces the problem to matrix factorization.
While intricate network architectures, e.g., convolutional and residual neural networks, pooling, autoencoders and generative adversarial neural networks are common in practice, they are usually designed in a way that facilitates training rather than making it difficult~\cite{Goodfellow2016_DeepLearning}.
We strengthen the result in~\cite{Abrahamsen2021_NeuralNetworks} by showing \ER-completeness for \emph{\fullyconnected} \twolayer neural networks.
This shows that \ER-hardness does not stem from one specifically chosen worst-case architecture, but is inherent in the neural network training problem itself.
Although a host of different architectures are used in practice, \fullyconnected \twolayer neural networks are arguably the most basic ones, and they are often part of more complicated network architectures~\cite{Goodfellow2016_DeepLearning}.
We show hardness even for the case of \fullyconnected \twolayer \ReLU neural networks with exactly two input and output dimensions.

Remarkably, with only one instead of two output dimensions, the problem is in \NP.
This follows from a combinatorial search algorithm by \citeauthor{Arora2018_Understanding}~\cite{Arora2018_Understanding}.
Our result explains why their algorithm was never successfully generalized to more complex network architectures: adding only a second output neuron significantly increases the computational complexity of the problem, from being contained in \NP{} to being \ER-complete.

To achieve this result, our reduction follows a completely novel approach compared to the reduction in~\cite{Abrahamsen2021_NeuralNetworks}.
Instead of encoding polynomial inequalities into an adversarial network architecture, we make use of the underlying geometry of the functions computed by \twolayer neural networks and utilize the fact that their different output dimensions have nonlinear dependencies.

\paragraph{Outline}
We start by formally introducing neural networks, the training problem, and the existential theory of the reals in \cref{sec:preliminaries}.
Thereafter, we present our main results in \cref{sec:results}.
In \cref{sec:discussion} we provide different perspectives on how to interpret our findings, including an in-depth discussion of the strengths and limitations.
We cover further related work in \cref{sec:related_work}.
We present the key ideas we use to prove \ER-hardness in \cref{sec:key_ideas}.
Finally, the complete proof details for \ER-containment, \ER-hardness, and algebraic universality are contained in \Cref{sec:membership,sec:hardness,sec:universality}, respectively.

\section{Preliminaries}
\label{sec:preliminaries}

\paragraph{Neural Networks}
All neural networks considered in this paper have a very simple architecture.
For ease of presentation, we do not define neural networks and their training problem in full generality here, but restrict ourselves to the simple architectures we need.

\begin{definition}
    A \emph{\fullyconnected \twolayer neural network}~$N = (S \,\dot{\cup}\, H \,\dot{\cup}\, T, E)$ is a directed acyclic graph (the \emph{architecture}) with real-valued edge weights.
    The vertices, called \emph{neurons}, are partitioned into the \emph{inputs}~$S$, the \emph{hidden neurons}~$H$ and the \emph{outputs}~$T$.
    All possible edges from~$S$ to~$H$, as well as all possible edges from~$H$ to~$T$ are present.
    Additionally, each hidden neuron has a real-valued \emph{bias} and an 
    \emph{activation function} $(\R \to \R)$.
\end{definition}

The probably most commonly used activation function~\cite{Arora2018_Understanding,Glorot2011_DeepSparse,Goodfellow2016_DeepLearning} is the \emph{rectified linear unit (ReLU)} defined as
\begin{align*}
    \ReLU \colon \R &\to \R \\
    x &\mapsto \max\{0,x\}
    \text{.}
\end{align*}
See \cref{fig:neural_networks_definition} for a small \fullyconnected \twolayer \ReLU neural network.

\begin{figure}[htb]
    \centering
    \includegraphics{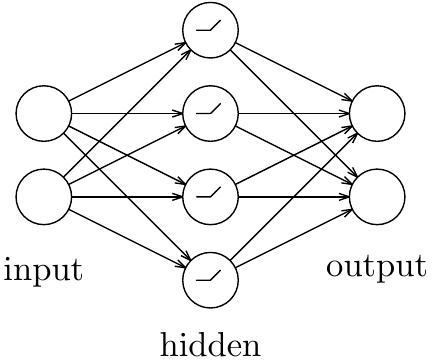}
    \caption{
        A \fullyconnected \twolayer neural network as studied in this paper.
        The symbol inside the hidden neurons expresses the \ReLU activation function.
    }
    \label{fig:neural_networks_definition}
\end{figure}

\medskip
\noindent
Given a neural network architecture~$N$ as defined above, let us fix an arbitrary ordering on~$S$ and~$T$.
Then~$N$ realizes a function~$f(\cdot, \Theta) \colon \R^{\abs{S}} \to \R^{\abs{T}}$, where~$\Theta$ denotes the weights and biases that parameterize the function.
For $x \in \R^{\abs{S}}$, we define~$f(\cdot, \Theta)$ inductively:
The $i$-th input neuron forwards the $i$-th component of~$x$ to all its outgoing neighbors.
Each hidden neuron forms the weighted sum over all incoming values, adds its bias, applies its activation function to this sum, and forwards it to all outgoing neighbors.
An output neuron also forms the weighted sum over all incoming values, but it neither adds a bias nor applies any activation function.

For our purposes, we define the following special case of \ERM, i.e., a restriction of the general neural network training problem.

\begin{definition}
    \label{def:trainNN}
    Training a \fullyconnected \twolayer neural network is the following decision problem, denoted by \trainNN:
    \begin{description}[nosep]
        \item[Input:]
        A~$5$-tuple $(N, \varphi, D, \gamma, c)$:
        \begin{itemize}[nosep]
            \item $N = (S \,\dot{\cup}\, H \,\dot{\cup}\, T, E)$ is the \fullyconnected \twolayer network architecture.
            \item $\varphi \colon \R \to \R$ is the activation function for all hidden neurons, computable in polynomial time on a \realRAM.
            \item $D \subseteq \Q^{\abs{S}} \times \Q^{\abs{T}}$ is the \emph{training data}, i.e., a set of~$n$ \emph{data points} of the form~$(x;y)$.
            Here,~$y$ is called the \emph{label} of the data point.
            \item $\gamma \in \Q_{\geq 0}$ is the \emph{target error}.
            \item $c \colon \R^{\abs{T}} \times \R^{\abs{T}} \to \R_{\geq 0}$ is an honest\footnote{A loss function is called \emph{honest} if it returns zero if and only if the data is fit exactly.} \emph{loss function}, computable in polynomial time on a \realRAM.
        \end{itemize}
        
        \item[Question:]
        Are there weights and biases~$\Theta$ such that
        \(
            \sum\limits_{\mathclap{(x;y) \in D}} c\bigl(f(x, \Theta), y\bigr) \leq \gamma
            \text{?}
        \)
    \end{description}
\end{definition}

Our \ER-hardness proof works for any fixed honest loss function, since our reduction only constructs instances with $\gamma=0$.
On the other hand, our \ER-membership proof only works for activation and loss functions that are polynomial-time computable on a \realRAM (see \cref{sec:membership} for the necessary details).

\paragraph{Existential Theory of the Reals}
The complexity class \ER has gained a lot of interest in recent years.
It is defined via its canonical complete problem \ETR (short for \emph{Existential Theory of the Reals}) and contains all problems that polynomial-time many-one reduce to it.
An \ETR instance consists of an integer~$n$ and a sentence of the form
\[
    \exists X_1, \ldots, X_n \in \R :
    \varphi(X_1, \ldots, X_n),
\]
where~$\varphi$ is a well-formed and quantifier-free formula consisting of polynomial equations and inequalities in the variables with integer coefficients encoded in binary, and the logical connectives $\{\land, \lor, \lnot\}$.
The goal is to decide whether this sentence is true.
As an example, consider the formula $\varphi(X,Y) :\equiv X^2 + Y^2 \leq 1 \land Y^2 \geq 2X^2 - 1$;
among (infinitely many) other solutions, $\varphi(0,0)$ evaluates to true, witnessing that this is a yes-instance of \ETR.
It is known that
\[
    \NP \subseteq \ER \subseteq \PSPACE
    \text{,}
\]
and it is widely believed that both inclusions are strict.
Here, the first inclusion follows because a \problemname{SAT} instance can easily be expressed as an equivalent \ETR instance~\cite{Shor1991_Stretchability}.
The second inclusion was proved by \citeauthor{Canny1988_PSPACE}~\cite{Canny1988_PSPACE}.

Note that the complexity of problems involving real numbers was studied in various contexts.
To avoid confusion, let us emphasize that the underlying machine model for \ER (over which sentences need to be decided and where reductions are performed in) is the standard binary \wordRAM (or equivalently, a Turing machine), just like for \P, \NP, and \PSPACE.
It is \emph{not} the \realRAM~\cite{Erickson2022_SmoothingTheGap} or the Blum-Shub-Smale model~\cite{Blum1989_ComputationOverTheReals}.

\section{Results}
\label{sec:results}

Our main result is the following \lcnamecref{thm:er_complete} establishing \ER-completeness of \trainNN even in very restricted cases:

\begin{theorem}
    \label{thm:er_complete}
    \trainNN is \ER-complete, even if
    \begin{itemize}[nosep]
        \item there are only two input neurons,
        \item there are only two output neurons,
        \item the number of data points is linear in the number of hidden neurons,
        \item the data has only~$13$ different labels,
        \item the target error is $\gamma=0$ and
        \item the \ReLU activation function is used.
    \end{itemize}
\end{theorem}

Let us note that the combination of all restrictions in \cref{thm:er_complete} describes a special case of \trainNN.
Proving this special case \ER-hard also implies \ER-hardness of the general version of \trainNN, and even its general case \ERM.
For example, \ER-hardness also holds for more input/output neurons, more data points, more labels, arbitrary~$\gamma \geq 0$, and other activation functions.

Additionally, our \ER-hardness reduction implies \emph{algebraic universality}:

\begin{theorem}
    \label{thm:universality}
    Let~$\alpha \in \R$ be an algebraic number.
    Then there exists an instance of \trainNN fulfilling all the restrictions in \Cref{thm:er_complete}, which has a solution with weights and biases~$\Theta$ from~$\Q[\alpha]$, but no solution when the weights and biases~$\Theta$ are restricted to a field~$\F$ not containing~$\alpha$.
\end{theorem}

Here, $\Q[\alpha]$ is the smallest field extension of $\Q$ containing $\alpha$. 
This means that there are training instances for which all global optima require irrational weights or biases, even if all data points are integral.

Let us note that algebraic universality and \ER-completeness are independent concepts:
While algebraic universality is known to hold for various \ER-complete problems~\cite{Abrahamsen2019_Toolbox}, this is not an automatism:
It cannot occur in problems where the solution space is open.
On the other hand, reductions to prove algebraic universality do not need to be in polynomial time, i.e., algebraic universality can be established without proving \ER-hardness at the same time.

\section{Discussion}
\label{sec:discussion}

There is already a wide body of literature about the computational hardness of \ERM, see \cref{sec:related_work} below.
In order to clarify our contribution, we use this \lcnamecref{sec:discussion} to discuss our results from various perspectives, pointing out strengths and limitations.

\paragraph{Implications of \ER-Completeness}
Our definition of \ER relies on \ETR as a canonical complete problem.
While \ETR is known to be \NP-hard and in \PSPACE~\cite{Canny1988_PSPACE, Shor1991_Stretchability}, its precise computational complexity is still unknown (and it is considered likely that it is neither \NP- nor \PSPACE-complete).
Therefore, proving a problem to be \ER-complete only tells us its difficulty relative to \ETR.
Nevertheless, proving \ER-completeness is a valuable contribution, even for problems that are known to be \NP-hard, because of the following implications.

As pointed out by \citeauthor{Schaefer2010_GeometryTopology}, \ER-completeness of a problem shifts the focus away from the problem itself to its underlying algebraic nature:
\enquote{Knowing that a problem is \ER-complete does not tell us more than that it is \NP-hard and in \PSPACE{} in terms of classical complexity, but it does tell us where to start the attack:~[...]
A solution will likely not come out of graph drawing or graph theory but out of a better understanding of real algebraic geometry and logic.}~\cite{Schaefer2010_GeometryTopology}

Knowing that a problem is \ER-complete also hints towards the algorithmic challenges that need to be overcome.
While many \NP-complete problems can be solved well in practice by extremely optimized off-the-shelf \SAT- or \MIP-solvers, no such general purpose tools are available for \ER-complete problems.
In fact, to the best of our knowledge, finding the optimal solution for any \ER-complete problem requires algorithms from real algebraic geometry, for example a decision procedure for existentially quantified first-order sentences (after reducing the problem to \ETR).
However, these algorithms are very inefficient (at least exponential in the number of variables) and therefore infeasible for large instances~\cite{Passmore2009_Experiments}.

However, let us stress that \ER-completeness does not rule out hope for good heuristics:
The \ER-complete art gallery problem can be solved well in practice using custom heuristics, some of which are also used in combination with integer programming solvers~\cite{Rezende2016_ArtGalleries}.
Under additional assumptions, performance guarantees can be proven~\cite{Hengeveld2021_ArtGallery}.
However, these heuristics are specifically tailored towards the art gallery problem.
Identifying reasonable assumptions for other \ER-complete problems in order to obtain good heuristics (possibly even with performance guarantees) is an important open question.

One meta-heuristic that can be used often to get \enquote{good} solutions for \ER-complete problems is gradient descent.
A prominent example in our context is training neural networks, where a bunch of different gradient descent variants powered by backpropagation are nowadays capable of training neural networks containing millions of neurons.
In general, we do not get any approximation or runtime guarantees when using gradient descent, but under the right additional assumptions proving such guarantees are sometimes possible~\cite{Brutzkus2017_OptimalGD}.
Gradient descent has also been applied to \ER-complete problems from other areas, for example \emph{graph drawing}~\cite{Ahmed2022_GraphDrawingGD}.
Still, these gradient descent approaches are tailored towards the specific problem at hand.
It would be very desirable to have general purpose solvers, similar to \SAT- or \MIP-solvers for problems in \NP.

\paragraph{Relation to Learning Theory}
In this paper, we purely focus on the computational complexity of \ERM, that is, minimizing the \emph{training error}.
In the machine learning practice, one usually desires to achieve low \emph{generalization error}, which means to use the training data to achieve good predictions on unseen test samples.

To formalize the concept of the generalization error, one needs to combine the computational aspect with a statistical one.
There are various models to do so in the literature, the most famous one being \emph{probably approximately correct} (PAC) learnability~\cite{Shalev2014_UnderstandingML,Valiant1984_Theory}.
While \ERM and learnability are two different questions, they are strongly intertwined; see \cref{sec:related_work} for related work in this context.
Despite the close connections between \ERM and learnability, to the best of our knowledge, the \ER-hardness of the former has no direct implications on the complexity of the latter.
Still, since \ERM is the most common learning paradigm in practice, our work is arguably also interesting in the context of learning theory.

\paragraph{Required Precision of Computation}
An implication of \ER-hardness of \trainNN is that for some instances every set of weights and biases exactly fitting the data needs large precision, actually a superpolynomial number of bits to be written down.
The algebraic universality of \trainNN (\cref{thm:universality}) strengthens this by showing that exact solutions require algebraic numbers.
This restricts the techniques one could use to obtain optimal weights and biases even further, as it rules out numerical approaches (even using arbitrary-precision arithmetic), and shows that symbolic computation is required.
In practice, we are often willing to accept small additive errors when computing~$f(\cdot, \Theta)$, and therefore do not require~$\Theta$ being of such high precision.
Rounding the weights and biases~$\Theta$ to the first \enquote{few} digits after the comma may be sufficient.
This might allow placing the problem of \emph{approximate} neural network training in \NP.
Yet, we are not aware of such a proof, and we consider it an interesting open question to establish this fact thoroughly.
Let us note that a similar phenomenon appears with many other \ER-complete problems~\cite{Erickson2022_SmoothingTheGap}:
While an exact solution~$x$ requires high precision, there is an approximate solution~$\tilde{x}$ close to~$x$ that needs only polynomial precision.
However, guessing the digits of the solution in binary is not a practical algorithm.
Moreover, historically, \ER-completeness seems to be a strong predictor that finding these approximate solutions is difficult in practice~\cite{Deligkas2022_ApproximatingETR, Erickson2022_SmoothingTheGap}.

\Citeauthor{Bienstock2023_LinearProgramming}~\cite{Bienstock2023_LinearProgramming} use the above idea to discretize the weights and biases to show that, in principle, arbitrary neural network architectures can be trained to approximate optimality via linear programs whose is size linear in the size of the data set, but exponential in the architecture size.
Let us emphasize that this does not imply \NP-membership of an approximate version of neural network training.

\paragraph{Number of Input Neurons}
In practice, neural networks are often trained on high dimensional data, thus having only two input neurons is even more restrictive than the practical setting.
Note that we easily obtain hardness for higher input dimensions by simply placing all data points of our reduction into a two-dimensional subspace.
The precise complexity of training \fullyconnected \twolayer neural networks with only one-dimensional input and multi-dimensional output remains unknown.
While this setting does not have practical relevance, we are still curious about this open question from a purely mathematical perspective.

\paragraph{Number of Output Neurons}
If there is only one output neuron instead of two, then the problem is known to be \NP-complete~\cite{Arora2018_Understanding, Froese2023_FixedDimension}.
Our reduction can easily be extended to the case with more than two output neurons by padding all output vectors with zeros.
Thus, the complexity classification is complete with respect to the number of output neurons.

\paragraph{Number of Hidden Neurons}
Consider a situation where the number~$m$ of hidden neurons is larger
than the number~$n$ of data points.
If there are no two contradicting data points~$(x_1; y_1)$ and~$(x_2; y_2)$ with $x_1 = x_2$ but $y_1 \neq y_2$, then we can always fit all data points exactly~\cite{Zhang2021_Understanding}.
Thus, we need at least a linear number of data points in terms of~$m$ for \ER-hardness.
Our result is (asymptotically) tight in this aspect.
Note that by adding additional data points, the ratio between~$n$ and~$m$ can be made arbitrarily large.
Thus, our reduction holds also for all settings in which~$m$ is (asymptotically) much smaller than~$n$.

\paragraph{Number of Output Labels}
The number of labels used in our reduction is just~$13$.
Requiring only a small constant number of different labels shows the relevance of our result to practice, where the number of data points often largely exceeds the number of labels, for instance, in classification tasks.

If all labels are contained in a one-dimensional affine subspace, then the problem is in \NP, as they can be projected down to one-dimensional labels and the problem can be solved with the algorithm by \citeauthor{Arora2018_Understanding}~\cite{Arora2018_Understanding}.
As any two labels span a one-dimensional affine subspace, the problem can only be \ER-hard for at least three affinely independent output labels.

We think it is not particularly interesting to close the gap between~$3$ and~$13$ output labels, but it would be interesting to investigate the complexity of the problem when output labels have more structure.
For example, in classification tasks one often uses \emph{one-hot encodings}, where the output dimension equals the number of classes and all labels have the form $(0, \ldots, 0, 1, 0, \ldots, 0)$.
Note that at least three output dimensions are needed in this case to obtain three different labels.

\paragraph{Target Error}
For simplicity, we only prove hardness for the case with target error~$\gamma = 0$.
However, it is generally not required to fit the data exactly in real-world applications.
It is not too difficult to see that we can modify the value of~$\gamma$ by adding inconsistent data points that can only be fit best in exactly one way.
The precise choice of these inconsistent data points heavily depends on the loss function.
In conclusion, for different values of~$\gamma$, the decision problem does not get easier.

\paragraph{Activation Function}
The \ReLU~activation function is currently the most commonly used activation function~\cite{Arora2018_Understanding, Glorot2011_DeepSparse, Goodfellow2016_DeepLearning}.
Our methods are adaptable to other piecewise linear activation functions, such as \emph{leaky \ReLU{}s}.
Having said that, our methods are \emph{not} applicable to other types of activation functions, such as \emph{Sigmoid}, \emph{soft \ReLU} or step functions.
We want to point out that \trainNN (and even \ERM) is in \NP{} if a step function is used as the activation function~\cite{Khalife2023_Threshold}.
Concerning the Sigmoid and soft \ReLU function, it is not even clear whether the \ERM is decidable, as trigonometric functions and exponential functions are not computable on the real RAM~\cite{Erickson2022_SmoothingTheGap, Richardson1969_Undecidable}.

\paragraph{Other Architectures}
We consider \fullyconnected \twolayer networks as the most important case, but we are also interested in \ER-hardness results for other network architectures.
Specifically, \fullyconnected three-layer neural networks and convolutional neural networks are interesting.
While it is hard to imagine that more complicated architectures are easier to train, a formal proof of this intuition would strengthen our result and show that \ER-completeness is a robust phenomenon, in other words, independent of a choice of a specific network type.

\paragraph{Lipschitz Continuity}
The set of data points created in the proof of \cref{thm:er_complete} is intuitively very tame.
Formally, this is captured by proving that for all yes-instances constructed by our reduction there exists a solution~$\Theta$ such that~$f(\cdot, \Theta)$ is Lipschitz continuous for a small Lipschitz constant~$L$.
Lipschitz continuity is also related to \emph{overfitting} and \emph{regularization}~\cite{Gouk2021_Regularisation}, the purpose of the latter being to prefer simpler functions over more complicated ones.
Being Lipschitz continuous with a small Lipschitz constant
essentially means that the function is relatively flat.
It is particularly remarkable that we can show hardness even for small Lipschitz constants, since Lipschitz continuity has been a crucial assumption in several recent results about training and learning \ReLU networks, for example in~\cite{Bienstock2023_LinearProgramming, Chen2022_LearningFTP, Goel2017_ReliablyLearning}.

\section{Related Work}
\label{sec:related_work}

\paragraph{Complexity of Neural Network Training}
It is well-known that minimizing the training error of a neural network is a computationally difficult problem for a large variety of activation functions and architectures~\cite{Shalev2014_UnderstandingML}.

Closest to our work is the recent \ER-completeness result by \citeauthor{Abrahamsen2021_NeuralNetworks} for \twolayer neural networks~\cite{Abrahamsen2021_NeuralNetworks}.
In contrast to our work, they use the identity activation function and rely on particularly difficult to train architectures, both qualities being uncommon in practice.
\citeauthor{Zhang1992_BSS}~\cite{Zhang1992_BSS} sketched a similar result already in \citeyear{Zhang1992_BSS}:
Training neural networks with \emph{real-valued} data points is $\NP_\R$-complete, again with the identity activation function and with an adversarial architecture.
$\NP_\R$ is a complexity class in the BSS-model of computation~\cite{Blum1989_ComputationOverTheReals}, but a suitable discretization (considering the so-called \emph{constant-free Boolean part}) yields \ER-completeness in today's language.

For \ReLU networks, \NP-hardness, parameterized hardness and inapproximability results have been established even for the simplest possible architecture consisting of only a single \ReLU neuron~\cite{Boob2022_Complexity, Dey2020_Approximation, Froese2022_Parameterized, Goel2021_Hardness}.
While all these results require non-constant input-dimension, \citeauthor{Froese2023_FixedDimension} show that it is still \NP-hard to train a \twolayer \ReLU network with two input neurons and one output neuron~\cite{Froese2023_FixedDimension}.
On the positive side, the seminal algorithm by \citeauthor{Arora2018_Understanding}~\cite{Arora2018_Understanding} solves \ERM for \twolayer \ReLU networks with one output neuron to global optimality, placing the problem in \NP.
It was later extended to a more general class of loss functions by \citeauthor{Froese2022_Parameterized}~\cite{Froese2022_Parameterized}.
The running time is exponential in the number of neurons in the hidden layer and in the input dimension, but polynomial in the number of data points if the former two parameters are considered to be constant.
This \NP-containment of \ERM with one-dimensional output is in sharp contrast to our \ER-completeness result for \trainNN with two-dimensional outputs.

While minimizing training and generalization errors are different problems, the hardness of the former also imposes challenges on the latter.
Strategies to circumvent hardness from the perspective of learning theory include allowing improper learners, restricting the type of weights allowed in a neural network, or imposing assumptions on the underlying distribution.
For example, \citeauthor{Chen2022_LearningFTP}~\cite{Chen2022_LearningFTP} show fixed-parameter tractability of learning a \ReLU network under several assumptions, including Gaussian data and Lipschitz continuity of the network.
We refer to~\cite{Bakshi2019_Learning, Chen2022_Hardness, Diakonikolas2020_Approximation, Goel2017_ReliablyLearning, Goel2019_Learning, Goel2018_Learning} as a non-exhaustive list of other results about (non-)learnability of \ReLU networks in different settings.

\paragraph{Expressivity of $\bm{{\ReLU}}$ Networks}
It is essential for our reduction to understand the classes of functions representable by certain \ReLU network architectures.
So-called \emph{universal approximation theorems} state that a single hidden layer (with arbitrary width) is already sufficient to approximate every continuous function on a bounded domain with arbitrary precision~\cite{Cybenko1989_Approximation, Hornik1991_Approximation}.
However, deeper networks require much fewer neurons to reach the same expressive power, yielding a potential theoretical explanation of the dominance of deep networks in practice~\cite{Arora2018_Understanding, Eldan2016_Power, Hanin2019_Universal, Hanin2018_Approximating, Liang2017_Deep, Nguyen2018_Neural, Raghu2017_Expressive, Safran2017_Depth, Telgarsky2016_Benefits, Yarotsky2017_Error}.
Other related work includes counting and bounding the number of linear regions~\cite{Hanin2019_Complexity, Montufar2022_Sharp, Montufar2014_Regions, Pascanu2014_Number, Raghu2017_Expressive, Serra2018_Bounding}, classifying the set of functions \emph{exactly} representable by different architectures~\cite{Arora2018_Understanding, Dereich2022_Minimal, Haase2023_LowerBound, Hertrich2021_Towards, Hertrich2021_ReLU, Mukherjee2017_Lower, Zhang2018_Tropical}, or analyzing the memorization capacity of \ReLU networks~\cite{Vardi2022_Optimal, Yun2019_Small, Zhang2021_Understanding}.
\citeauthor{Huchette2023_Survey}~\cite{Huchette2023_Survey} provide a survey on the interactions of neural networks and polyhedral geometry, including implications on training, verification, and expressivity.

\paragraph{Existential Theory of the Reals}
The complexity class \ER gains its importance from numerous important algorithmic problems that have been shown to be complete for this class in recent years.
The name \ER was introduced by \citeauthor{Schaefer2010_GeometryTopology}~\cite{Schaefer2010_GeometryTopology} who also pointed out that several \NP-hardness reductions from the literature actually implied \ER-hardness.
For this reason, several important \ER-completeness results were obtained before the need for a dedicated complexity class became apparent.

Common features of \ER-complete problems are their continuous solution space and the nonlinear relations between their variables.
Important \ER-completeness results include the realizability of abstract order types~\cite{Mnev1988_UniversalityTheorem, Shor1991_Stretchability} and geometric linkages~\cite{Schaefer2013_Realizability}, as well as the recognition of many types of geometric intersection graphs~\cite{Cardinal2018_Intersection,Kang2012_Sphere,Kratochvil1994_IntersectionGraphs, Matousek2014_IntersectionGraphsER,McDiarmid2013_DiskSegmentGraphs}.
More results appeared in the graph drawing community~\cite{Dobbins2023_AreaUniversality, Jungeblut2024_Lombardi, Lubiw2022_DrawingInPolygonialRegion, Schaefer2021_FixedK, Schaefer2023_RAC}, regarding polytopes~\cite{Richter1995_Polytopes}, the study of Nash-equilibria~\cite{Berthelsen2022_MultiPlayerNash, Bilo2021_Nash, Garg2018_MultiPlayer, Schaefer2017_FixedPointsNash}, matrix factorization~\cite{Chistikov2016_Matrix, Schaefer2018_TensorRank, Shitov2017_PSMatrixFactorization}, or continuous constraint satisfaction problems~\cite{Miltzow2022_ContinuousCSP}.
In computational geometry, we would like to mention the art gallery problem~\cite{Abrahamsen2022_ArtGallery, Stade2023_ArtGallery} and covering polygons with convex polygons~\cite{Abrahamsen2022_Covering}.

Recently, the community started to pay more attention to higher levels of the \emph{real polynomial hierarchy}, which also capture several interesting algorithmic problems~\cite{Blanc2021_ESS, Burgisser2009_ExoticQuantifiers, DCosta2021_EscapeProblem, Dobbins2023_AreaUniversality, Jungeblut2024_Hausdorff, Schaefer2023_RealHierarchy}.

\section{Proof Ideas}
\label{sec:key_ideas}

To prove \ER-hardness, we reduce from \ETRINV, a restricted version of \ETR.
An instance of \ETRINV consists of real variables and a conjunction of constraints in these variables.
Each constraint is either an addition constraint $X + Y = Z$, or an inversion constraint $X \cdot Y = 1$.
Given such an instance, we construct a \trainNN instance that models the variables, and in which exact fitting of all the data points corresponds to simultaneously satisfying all addition and inversion constraints.

\paragraph{Variables}
A natural candidate for encoding variables are the weights and biases of the neural network.
However, those did not prove to be suitable for our purposes.
The main problem with using the parameters of the neural network as variables is that the same function can be computed by many neural networks with different combinations of these parameters.
We are not aware of an easy way to normalize the parameters.

To circumvent this issue, we work with the functions representable by \fullyconnected \twolayer neural networks directly.
We frequently make use of the geometry of the plots of these functions.
For now, it is only important to understand that each hidden \ReLU neuron encodes a \CPWL function with exactly two pieces (both of constant gradient).
These two pieces are separated by a so-called \emph{breakline}.
Now, if we have~$m$ hidden neurons, their individually encoded functions add up such that the function computed by the whole neural network is a \CPWL function with at most~$m$ breaklines.
All cells of the function between the breaklines have constant gradient.

To keep things simple for now, let us first consider a neural network with only one input and one output neuron.
We place a series of data points~$(x_i; y_i) \in \R\times\R$ as seen in \cref{fig:key_ideas_variable_gadget}.
All \CPWL functions~$f(\cdot,\Theta)$ computed by a neural network with only four hidden neurons (i.e., only four breaklines and therefore at most five pieces) that fit these data points exactly must be very similar.
In fact, they can only differ in one degree of freedom, namely the slope of the piece going through the data point~$p$.
In our construction, this slope represents the value of a variable.
The whole set of data points enforcing this configuration is called a \emph{variable gadget}.

\begin{figure}[htb]
    \centering
    \includegraphics[page=1]{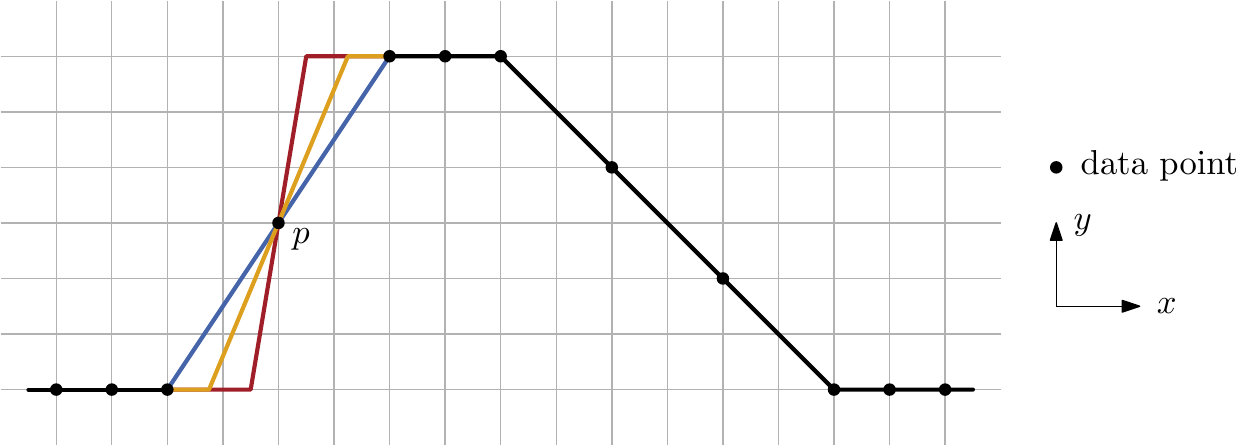}
    \caption{
        The value of~$f(\cdot, \Theta)$ is fixed (black part), except for the segment through data point~$p$.
        The red, orange and blue segments are just three out of uncountably many possibilities.
        Its slope can be used to encode a real-valued variable.
    }
    \label{fig:key_ideas_variable_gadget}
\end{figure}

\paragraph{Linear Dependencies}
The key insight for encoding constraints between variables is that we can relate the values of several variable gadgets by a data point:
By placing a data point~$p$ at a location where several variable gadgets overlap, each of the variable gadgets contributes its part towards fitting~$p$.
The exact contribution of each variable gadget depends on its slope.
Consequently, if one variable gadget contributes more, the others have to lower their contribution by the same amount.
This enforces linear dependencies between different variable gadgets and can be used to design \emph{addition} and \emph{copy gadgets}.

We need a second input dimension in order to intersect multiple variable gadgets.
This extends each variable gadget into a \emph{stripe} in~$\R^2$, with \cref{fig:key_ideas_variable_gadget} showing only an orthogonal cross-section of this stripe.
See \cref{fig:key_ideas_intersecting_variable_gadgets} for two intersecting variable gadgets.
Much of the technical difficulties lie in the subtleties to enforce the presence of multiple (possibly intersecting) gadgets using a finite number of data points.

\begin{figure}[htb]
    \centering
    \includegraphics[page=2]{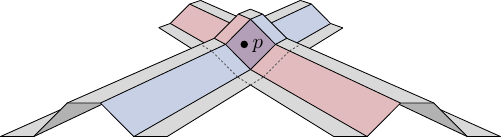}
    \caption{
        Two intersecting variable gadgets.
        The slopes of the blue and the red region encode the values.
        Point~$p$ lies in the intersection of both and can encode a linear relationship between them.
    }
    \label{fig:key_ideas_intersecting_variable_gadgets}
\end{figure}

An addition gadget encodes a linear relation between three variables.
In~$\R^2$, three lines (or stripes) usually do not intersect in a single point.
Thus, we have to carefully place the gadgets to guarantee such a single intersection point.
To this end, we create copies of the involved variable gadgets (copying is a linear relation between just two variables, thus \enquote{easy}).
These copies can then be positioned (almost) freely.

\paragraph{Inversion}
We are not able to encode nonlinear constraints within only a single output dimension~\cite{Arora2018_Understanding}.
By adding a second output dimension, the neural network now represents two functions~$f^1(\cdot, \Theta)$ and~$f^2(\cdot, \Theta)$.
Consequently, we are allowed to use data points with two different output labels, one for each output dimension.

One important observation is that the locations of the breaklines of $f = f(\cdot, \Theta) = (f^1(\cdot, \Theta), f^2(\cdot, \Theta))$ are independent of the weights of the edges in the second layer of the neural network.
Thus, both functions~$f^1$ and~$f^2$ have the same breaklines.
Still, setting some weights to zero may \emph{erase} a breakline in one of the functions.

An \emph{inversion gadget} (realizing the constraint $X \cdot Y = 1)$ also corresponds to a stripe in~$\R^2$.
For simplicity, we only show a cross-section here, see \cref{fig:key_ideas_inversion}.
In each output dimension individually, the inversion gadget acts exactly like a variable gadget.
The inversion gadget can therefore be understood as a variable gadget that carries two values.

We prove that by allowing only five breaklines in total, a function~$f$ can only fit all data points exactly if~$f^1$ and~$f^2$ share three of their four breaklines (while both having one \enquote{exclusive} breakline each, which is erased in the other dimension).
This enforces a nonlinear dependency between the slopes of~$f^1$ and~$f^2$.
By choosing the right parameters, this nonlinear relation models an inversion constraint.

\begin{figure}[hbt]
    \centering
    \includegraphics[page=3]{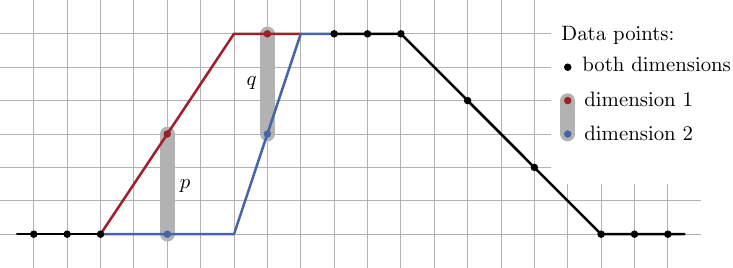}
    \caption{
        Data points~$p$ and~$q$ have different labels in the two output dimensions, enforcing that the slopes of the red and the blue pieces are related via a nonlinear dependency.
    }
    \label{fig:key_ideas_inversion}
\end{figure}

\paragraph{Reduction}
Let us illustrate the reduction by giving a simple example.
Note that this is not yet the complete picture.
We start with an \ETRINV instance, for example, deciding whether the following sentence
\begin{align*}
    \exists X_1, X_2, X_3, X_4 \in \R \colon
    &(X_1 + X_2 = X_3) \land
    (X_1 + X_3 = X_4)~\land \\
    &(X_1 \cdot X_4 = 1) \land
    (X_4 \cdot X_3 = 1)
\end{align*}
is true.
This instance has four variables~$X_1, X_2, X_3, X_4$ and four constraints: two additions and two inversions.
Recall that every gadget corresponds to a stripe in the input space~$\R^2$.
See \cref{fig:global_arrangement_intro} for the following construction (the stripes are drawn as lines for better readability).

\begin{itemize}
    \item We add a variable gadget for each of the variable.
    All of these are placed such that their corresponding stripes are parallel and do not overlap, see the horizontal lines in \cref{fig:global_arrangement_intro}.
    
    \item We introduce three more variable gadgets for each addition constraint, one per involved variable.
    These are placed such that they have a common intersection point while also intersecting their corresponding variable gadget.
    In \cref{fig:global_arrangement_intro}, see the two bundles to the left.
    A data point at the triple intersection enforces the addition constraint, while data points labelled~$\bullet_=$ encode that the values of the two intersecting variable gadgets are equal.
    
    \item Lastly, we add an inversion gadget for each inversion constraint and place it such that it intersects the variable gadgets of the two involved variables.
    See the two dashed lines in \cref{fig:global_arrangement_intro}.
    Data points labelled~$\bullet^{=_1}$ ($\bullet^{=_2}$) enforce equality only in the first (second) output dimension.
\end{itemize}

\begin{figure}[htb]
    \centering
    \includegraphics[page=1]{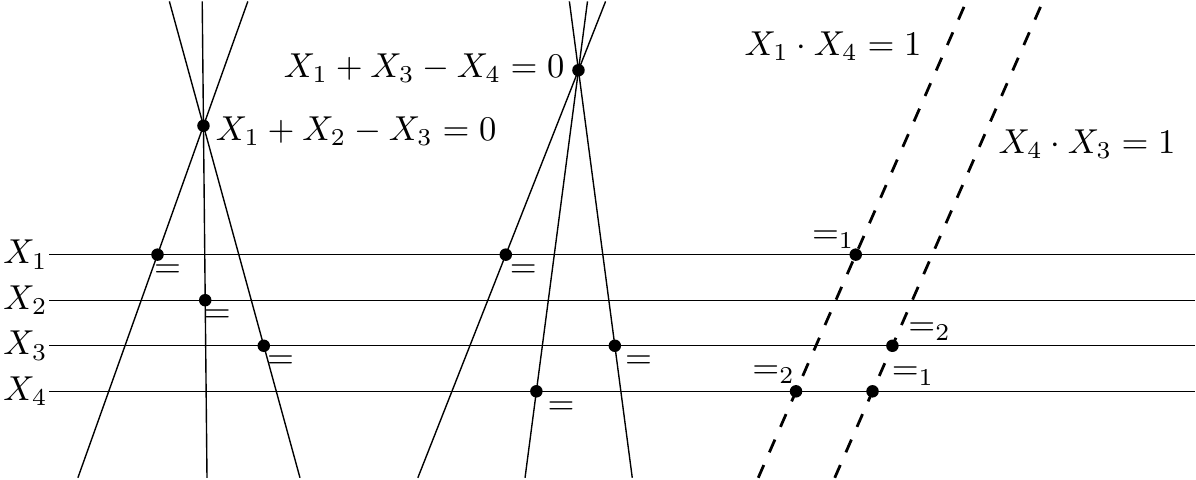}
    \caption{Overview of the global arrangement of the gadgets.}
    \label{fig:global_arrangement_intro}
\end{figure}

To see that the above reduction is correct, assume first that the \ETRINV instance is true.
Then there exist real values for the four variables and these values can be used as the slopes of their corresponding variable gadgets.
By the correctness of the individual gadgets (that we prove below), it follows that each data point is fit exactly.

Conversely, if all data points are fit exactly, then the correctness of the gadgets implies that the slopes of the variable gadgets give a solution to the \ETRINV instance.

\section{\texorpdfstring{$\bm{{\ER}}$}{ER}-Membership}
\label{sec:membership}

\ER-membership is already proven by \citeauthor{Abrahamsen2021_NeuralNetworks}:

\begin{proposition}[{\cite[Section~$2$]{Abrahamsen2021_NeuralNetworks}}]
    \label{prop:membership}
    $\trainNN \in \ER$.
\end{proposition}

For the sake of completeness, while not being too repetitive, we shortly summarize their argument:
\ER-membership is shown by describing a so-called \emph{polynomial-time real verification algorithm} (see~\cite{Erickson2022_SmoothingTheGap} for the formal details).
The input of such an algorithm is a \trainNN instance~$I$, as well as a witness~$\Theta$ consisting of real-valued weights and biases.
Instance~$I$ consists of a network architecture, data points~$D$ and a target error~$\gamma$.
The algorithm has to verify that the neural network parameterized by~$\Theta$ fits all data points in~$D$ with a total error at most~$\gamma$.
The underlying model of computation is the \emph{real RAM}, that is, an extension of the classical word RAM by registers that can store arbitrary real numbers.
Arithmetic operations (\enquote{$+$}, \enquote{$-$}, \enquote{$\cdot$}, \enquote{$\div$}) on real registers take constant time.

To achieve this, the real verification algorithm loops over all data points in~$D$ and evaluates the function realized by the neural network for each of them individually.
As each hidden neuron uses a polynomial-time computable activation function, each such evaluation takes only polynomial time in the size of the network.
\Cref{prop:membership} follows, since we are also guaranteed that the loss function can be computed in polynomial time on a real RAM.

\section{\texorpdfstring{$\bm{{\ER}}$}{ER}-Hardness}
\label{sec:hardness}

This \lcnamecref{sec:hardness} is devoted to proving \ER-hardness of \trainNN.
Our reduction is mostly geometric, so we start by reviewing the underlying geometry of \twolayer neural networks in \cref{sec:geometry_of_NN}.
This is followed by a high-level overview of the reduction in \cref{sec:reduction_overview}, before we describe the gadgets in detail in \cref{sec:gadgets}.
Finally, in \cref{sec:global_layout}, we combine the gadgets into the proof of \cref{thm:er_complete}.

\subsection{Geometry of Two-Layer Neural Networks}
\label{sec:geometry_of_NN}

Our reduction constructs a neural network that has just two input neurons, two output neurons, and~$m$ hidden neurons.
Thus, for given weights and biases~$\Theta$, it realizes a function $f(\cdot,\Theta) \colon \R^{2} \to \R^{2}$.
In this \lcnamecref{sec:geometry_of_NN}, we build a geometric understanding of~$f(\cdot, \Theta)$, in particular, we study the geometry of the plot of~$f(\cdot, \Theta)$.
For further results in this direction, we point the interested reader to~\cite{Arora2018_Understanding, Dereich2022_Minimal, Hertrich2021_Towards, Mukherjee2017_Lower, Zhang2018_Tropical} that investigate the set of functions exactly represented by different architectures of \ReLU networks.

The $i$-th hidden \ReLU neuron~$v_i$ realizes a function
\begin{align*}
    f_i \colon \R^2 &\to \R \\
    (x_1, x_2) &\mapsto \ReLU(a_{1,i}x_1 + a_{2,i}x_2 + b_i)
    \text{,}
\end{align*}
where~$a_{1,i}$ and~$a_{2,i}$ are the edge weights from the first and second input neuron to~$v_i$ and~$b_i$ is its bias.
Note that~$f_i$ is a \CPWL function:
If~$a_{1,i} = a_{2,i} = 0$, then~$f_i$ is constant, $f_i = \ReLU(b_i) = \max\{b_i, 0\}$.
Otherwise, the domain~$\R^2$ is partitioned into two half-planes, touching along a so-called \emph{breakline} given by the equation $a_{1,i}x_1 + a_{2,i}x_2 + b_i = 0$.
The two half-planes are (see \cref{fig:activeregion})
\begin{itemize}
    \item the \emph{inactive region} \(
        \bigl\{
            (x_1,x_2) \subseteq \R^2 \mid
            a_{1,i}x_1 + a_{2,i}x_2 + b_i \leq 0
        \bigr\}
    \), in which~$f_i$ is constantly~$0$, and
    
    \item the \emph{active region} \(
        \bigl\{
            (x_1,x_2) \subseteq \R^2 \mid
            a_{1,i}x_1 + a_{2,i}x_2 + b_i > 0
        \bigr\}
    \), in which~$f_i$ is positive and has a constant gradient.
\end{itemize}

\begin{figure}[htb]
    \centering
    \includegraphics{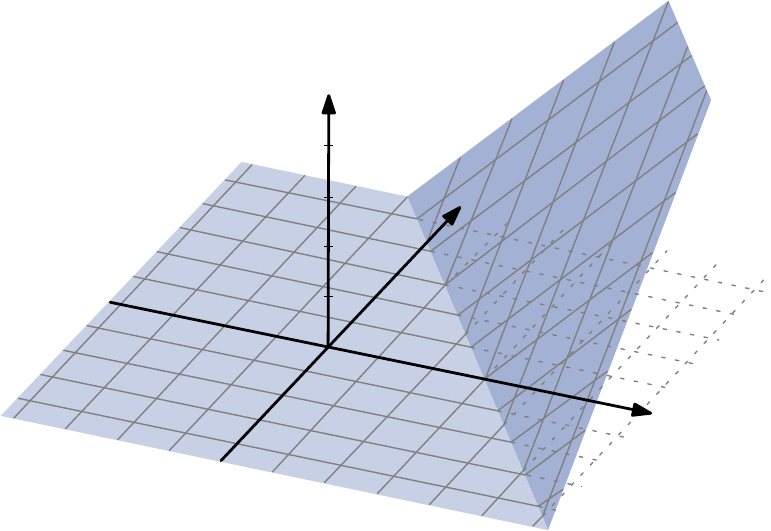}
    \caption{
        A \CPWL function computed by a hidden \ReLU neuron.
        It has exactly one breakline;
        the flat part is the inactive region, whereas the sloped part is the active region.
    }
    \label{fig:activeregion}
\end{figure}

Now let~$c_{i,1}$ and~$c_{i,2}$ be the weights of the edges connecting~$v_i$ with the first and second output neuron, and let~$f(\cdot, \Theta) = (f^1(\cdot, \Theta), f^2(\cdot, \Theta))$.
For~$j \in \{1,2\}$, the function~$f^j(\cdot, \Theta) = \sum_{i = 1}^m c_{i,j} \cdot f_i(\cdot, \Theta)$ is a weighted linear combination of the functions computed at the hidden neurons.
We make three observations:

\begin{itemize}
    \item Each function computed by a hidden \ReLU neuron has at most one breakline.
    Thus, the domain of~$f^j(\cdot, \Theta)$ is partitioned into the cells of a line arrangement containing at most~$m$ breaklines.
    Apart from that,~$f^j(\cdot, \Theta)$ has a constant gradient inside each cell.

    \item Let~$b_i$ be the breakline produced by a hidden neuron~$v_i$ in~$f^j(\cdot, \Theta)$.
    Its position is solely determined by~$a_{\cdot,i}$ and~$b_i$.
    In particular, it is independent of~$c_{i,j}$.
    Thus, the sets of breaklines of~$f^1(\cdot, \Theta)$ and~$f^2(\cdot, \Theta)$ are both subsets of the same set of (at most~$m$) breaklines determined by the hidden neurons.

    \item Even if all~$f_i(\cdot, \Theta)$ have a breakline, their sum $f^j(\cdot, \Theta)$ at each output neuron might have fewer breaklines:
    It is possible for a breakline to be \emph{erased} by setting $c_{i,j} = 0$.
    Other possibilities are that several breaklines contributed by different hidden neurons cancel each other (producing no breakline) or lie on top of each other (combining multiple breaklines into one).
    In our reduction we deliberately erase some breaklines in some output dimensions, i.e., we make use of the $c_{i,j} = 0$ trick.
    However, we avoid the other two cases of breaklines combining/canceling.
\end{itemize}

Combining above observations yields a stronger statement:
For each output neuron and breakline~$\ell$, the change of gradient of~$f^j(\cdot, \Theta)$ along~$\ell$ is constant (see also~\cite{Dereich2022_Minimal}).
Based on this, we distinguish two types of breaklines:

\begin{definition}
    A breakline~$\ell$ is \emph{concave} (\emph{convex}) in~$f^j(\cdot,\Theta)$ if the restriction of~$f^j(\cdot,\Theta)$ to any two neighboring cells separated by~$\ell$ is concave (convex).

    The \emph{type} of a breakline is a tuple $(t_1,t_2) \in \{\wedge, 0, \vee\}^2$ describing whether the breakline is concave~$(\wedge)$, erased~$(0)$, or convex~$(\vee)$ in~$f^1(\cdot,\Theta)$ and~$f^2(\cdot,\Theta)$, respectively.
\end{definition}

By now, we gained a geometric understanding of~$f(\cdot, \Theta)$, the \CPWL function computed by a \ReLU neural network with two input and two output neurons.
However, not every \CPWL function can be computed by such a neural network.
For the correctness of our reduction, we need a sufficient condition for this:

\begin{lemma}
    \label{lem:fittableCPWL}
    A \CPWL function~$f \colon \R^2 \to \R^2$ whose breaklines form a line arrangement with~$m$ lines can be realized by a \fullyconnected \twolayer neural network with~$m$ hidden neurons if the following two conditions hold:
    \begin{itemize}
        \item In at least one cell of~$\calL$ the value of~$f$ is constantly~$(0,0)$.
        \item For each breakline~$\ell \in \calL$, the change of the gradient of~$f$ along~$\ell$ is constant in both output dimensions.
    \end{itemize}
\end{lemma}

\begin{proof}
    We can use the following construction:
    Add one hidden neuron per breakline, oriented such that its inactive region is the half-plane containing the $(0,0)$-cell.
    The position solely depends on the weights of the first layer and the bias.
    The weights of the second layer are then chosen to produce the right change of gradient in each output dimension.
    It is easy to see that the sum of all these neurons computes~$f$.
\end{proof}

We refer to~\cite{Dereich2022_Minimal} for a precise characterization of the functions representable by \twolayer neural networks with~$m$ hidden neurons.

\subsection{Preparing the Reduction}
\label{sec:reduction_overview}

We show \ER-hardness of \trainNN by giving a polynomial-time reduction from \ETRINV to \trainNN.
\ETRINV is a variant of \ETR that is frequently used as a starting point for \ER-hardness proofs in the literature~\cite{Abrahamsen2022_ArtGallery, Abrahamsen2021_NeuralNetworks, Dobbins2023_AreaUniversality, Lubiw2022_DrawingInPolygonialRegion}.

Formally, \ETRINV is a special case of \ETR in which the quantifier-free part~$\varphi$ of the input sentence \(
    \Phi \dequiv
    \exists X_1, \ldots, X_n \in \R \colon
    \varphi(X_1, \ldots, X_n)
\) is a conjunction (only~$\land$ is allowed) of constraints, each of which is either of the form~$X + Y = Z$ or~$X \cdot Y = 1$.

A promise version of \ETRINV allows us to assume that there is either no solution or a solution with all variables in~$\bigl[\frac{1}{2},2\bigr]$.
Our reduction starts from this promise version.

\begin{theorem}[{\cite[Theorem 3.2]{Abrahamsen2022_ArtGallery}}]
    \label{thm:ETRINV}
    \ETRINV is \ER-complete.
\end{theorem}

Furthermore, \ETRINV exhibits the same algebraic universality that we seek for \trainNN:

\begin{theorem}[\cite{Abrahamsen2019_Toolbox}]
    \label{thm:ETRINV_universality}
    Let~$\alpha$ be an algebraic number.
    Then there exists an instance of \ETRINV, which has a solution in~$\Q[\alpha]$, but no solution when the variables are restricted to a field~$\F$ that does not contain~$\alpha$.
\end{theorem}

The reduction starts with an \ETRINV instance~$\Phi$ and outputs an integer~$m$ and a set of~$n$ data points such that there is a \fullyconnected \twolayer \ReLU neural network~$N$ with~$m$ hidden neurons exactly fitting all data points ($\gamma = 0$) if and only if~$\Phi$ is true.
Recall that the neural network~$N$ defines a \CPWL function~$f(\cdot, \Theta) \colon \R^2 \to \R^2$.

We define several \emph{gadgets} representing the variables as well as the linear and inversion constraints of the \ETRINV instance~$\Phi$.
Strictly speaking, a gadget is defined by a set of data points that need to be fit exactly.
These data points serve two tasks:
Firstly, most of the data points are used to enforce that~$f(\cdot, \Theta)$ has~$m$ breaklines with predefined orientations and at almost predefined positions.
Secondly, the remaining data points enforce relationships between the exact positions of different breaklines.

Globally, our construction yields $f(x, \Theta) = (0,0)$ for \enquote{most}~$x \in \R^2$.
Each gadget consists of a constant number of parallel breaklines (enforced by data points) that lie in a \emph{stripe} of constant width in~$\R^2$.
The value of~$f(\cdot, \Theta)$ may be non-zero only within these stripes.
The \enquote{semantics} of a gadget\footnote{
    For a variable gadget, its \enquote{semantics} is the real number represented by it.
} is fully determined by the distances between its parallel breaklines.
Thus, each gadget can be translated and rotated arbitrarily without affecting its semantics.

\paragraph{Abstractions}
Describing all gadgets purely by their data points is tedious and obscures the relatively simple geometry enforced by these data points.
We therefore introduce two additional constructs, namely \emph{data lines} and \emph{weak data points}, that simplify the presentation.
In particular, data lines impose breaklines, which in turn are needed to define gadgets.
Weak data points allow us to have features that are only active in one output dimension.
How these constructs can be realized with carefully placed data points is deferred to \cref{sec:lower_bound_gadget,sec:data_lines_to_data_points}, after we have introduced all other gadgets.

\begin{itemize}
    \item A \emph{data line}~$(\ell; y)$ consists of a line~$\ell \subseteq \R^2$ and a label~$y \in \R^2$.
    We say that a data line is fit if~$f(\ell, \Theta) = \{y\}$, i.e., the neural network maps every point on it to~$y$.
    
    As soon as we consider several gadgets, their corresponding stripes in~$\R^2$ might intersect.
    We do not require that the data lines are fit correctly inside these intersections.
    This is justified because, as we are going to see below, each data line is realized by finitely many data points on it.
    We make sure that their coordinates do not lie in any of the intersections.

    \item A \emph{weak data point} relaxes the notion of a regular data point and prescribes only a lower bound on the label.
    For example, we denote by~$(x; y_1, \geq y_2)$ that $f^1(x, \Theta) = y_1$ and $f^2(x, \Theta) \geq y_2$.
    Weak data points can have such an inequality label in the first, the second, or both output dimensions.
\end{itemize}

\subsection{Gadgets and Constraints}
\label{sec:gadgets}

We describe all gadgets in isolation first.
The interaction of two or more gadgets is considered only where it is necessary.
In particular, we assume that~$f(x, \Theta)$ is constantly zero for~$x \in \R^2$ outside the outermost breaklines enforced by each gadget.
After all gadgets have been introduced, we describe the global arrangement of the gadgets in \cref{sec:global_layout}.
Recall that, since each gadget can be freely translated and rotated, we can describe the positions of all its data lines and (weak) data points relative to each other.
Since all data lines of a gadget are parallel, we can describe their relative positions solely by the distance between them.

Not all gadgets make use of the two output dimensions.
Some gadgets have the same labels in both output dimensions for all of their data lines, and thus look the same in both output dimensions.
For these gadgets, we simplify the usual notation of $(y_1,y_2) \in \R^2$ to single-valued labels $y \in \R$.
In our figures, data points and functions looking the same in both output dimensions are drawn in black, while features only occurring in one dimension are drawn in different colors to distinguish them from each other.

To clarify our exposition, let us define some terms we will use.
Let $(\ell_1; y_i), \ldots, (\ell_k; y_k)$ be parallel data lines.
Further, let~$\ell \subseteq \R^2$ be an oriented line intersecting all~$\ell_i$.
Without loss of generality, we assume that~$\ell$ intersects~$\ell_i$ before~$\ell_j$ if and only if~$i < j$.
A \emph{cross-section} through $(\ell_1; y_i), \ldots, (\ell_k; y_k)$ is defined as follows:
For each data line~$(\ell_i; y_i)$, the cross-section contains a data point~$p_i = (x_i; y_i) \in \R \times \R^2$, where~$x_i$ is the oriented distance between the intersections of~$\ell_1$ and~$\ell_i$ with~$\ell$.
Two data points~$p_i$ and~$p_j$ in the cross-section are \emph{consecutive} if $\abs{i - j} = 1$.
If~$\ell$ is perpendicular to all~$\ell_i$, then the cross-section is \emph{orthogonal}.
The intersection of a breakline with~$\ell$ is a \emph{breakpoint}.

We draw cross-sections by projecting a data point~$(x_i; y_i) \in \R \times \R^2$ into a two-dimensional coordinate system, marking~$x_i$ along the abscissa and~$y_i$ along the ordinate.
If a~$y_i$ behaves differently in the two output dimensions, then we draw it twice and distinguish the two dimensions by color.

\begin{observation}
    \label{obs:breakpoints}
    Let~$f$ be a \CPWL function exactly fitting three consecutive data points~$p_i$, $p_{i+1}$ and~$p_{i+2}$ in a cross-section of a gadget.
    The following holds for each output dimension:
    \begin{enumerate}[label=(\roman*)]
        \item\label{obs:collinear_breakpoints}
        If the three points are collinear, then~$f$ has either no breakpoint strictly between~$p_i$ and~$p_{i+2}$, or at least two.

        \item\label{obs:non_collinear_breakpoints}
        If the three points are not collinear, then~$f$ has some breakpoint~$b$ strictly between~$p_i$ and~$p_{i+2}$.
        Furthermore, if~$p_{i+2}$ lies above (below) the ray from~$p_i$ through~$p_{i+1}$, then~$b$ is convex (concave).
    \end{enumerate}
\end{observation}

\Cref{obs:breakpoints} is the key to prove that data lines enforce breaklines of a certain type, with a prescribed orientation and (almost) fixed position.
It is illustrated in \cref{fig:breakpoints}.

\begin{figure}[htb]
    \begin{subfigure}[t]{.45\linewidth}
	\centering
	\includegraphics[page=1]{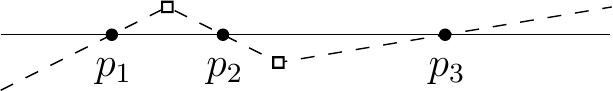}
	  \caption{If the points are collinear, then there is either no breakpoint or there are at least two.}
	  \label{fig:collinear_breakpoints}
    \end{subfigure}
    \hfill
    \begin{subfigure}[t]{.45\linewidth}
        \centering
	\includegraphics[page=2]{figs/breaklines.pdf}
	\caption{If the points are not collinear, then we need a breakpoint of a certain type, here convex.}
	\label{fig:non_collinear_breakpoints}
    \end{subfigure}
    \caption{Three consecutive points~$p_1$, $p_2$ and~$p_3$ in a cross-section and possible \CPWL functions fitting them (solid and dashed).}
    \label{fig:breakpoints}
\end{figure}

\subsubsection{Variable Gadget}

A \emph{variable gadget} consists of twelve parallel data lines~$\ell_1, \ldots, \ell_{12}$, numbered from one side to the other, in the figures from left to right.
Further, there is a weak data point~$q$ between~$\ell_3$ and~$\ell_4$.
For all of these, the following table lists their relative distance to~$\ell_1$ and their label (note that both output dimensions have the same label):
\begin{center}
    \begin{tabular}{lccccccccccccc}
        \toprule
        & $\ell_1$ & $\ell_2$ & $\ell_3$ & q & $\ell_4$ & $\ell_5$ & $\ell_6$ & $\ell_7$ & $\ell_8$ & $\ell_9$ & $\ell_{10}$ & $\ell_{11}$ & $\ell_{12}$ \\
        \midrule
        distance to $\ell_1$ & $0$ & $1$ & $2$ & $3 + \frac{2}{3}$ & $4$ & $6$ & $7$ & $8$ & $10$ & $12$ & $14$ & $15$ & $16$ \\
        label & $0$ & $0$ & $0$ & $\geq 2$ & $3$ & $6$ & $6$ & $6$ & $4$ & $2$ & $0$ & $0$ & $0$ \\
        \bottomrule
    \end{tabular}
\end{center}
See \cref{fig:variable_gadget} for an orthogonal cross-section through a variable gadget.

\begin{figure}[htb]
    \centering
    \includegraphics{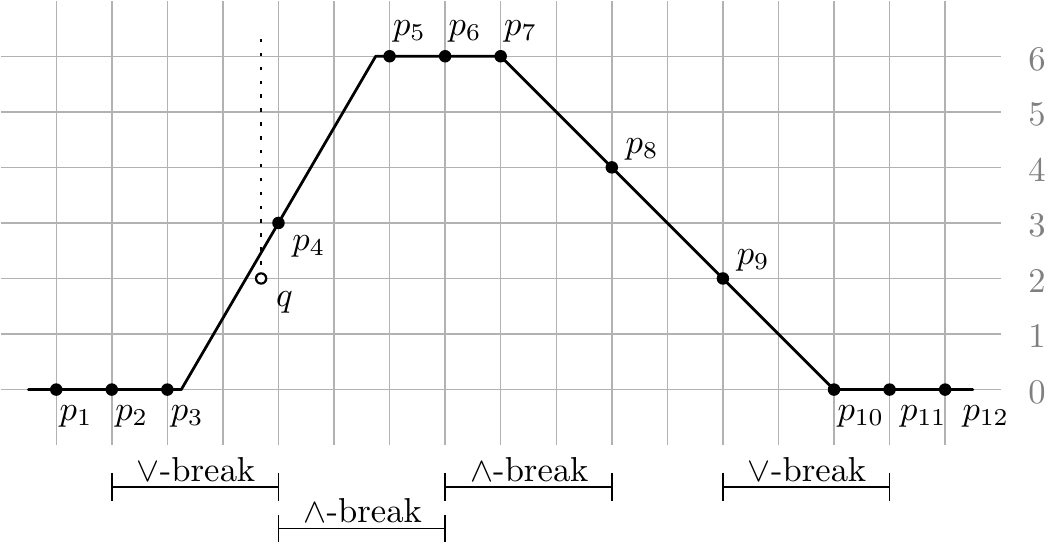}
    \caption{
        Orthogonal cross-section of a variable gadget.
        The bars below the cross-section indicate non-collinear triples used in the proof of \cref{lem:variable_gadget}.
        For example, there needs to be a convex breakpoint strictly between~$p_2$ and~$p_4$.
    }
    \label{fig:variable_gadget}
\end{figure}

\begin{lemma}
    \label{lem:variable_gadget}
    A \CPWL function~$f$ that fits~$\ell_1, \ldots, \ell_{12}$ and~$q$ exactly must have at least four breaklines.
    If it has exactly four breaklines, then they must all be parallel to the data lines.
    In this case, let~$b_1$, $b_2$, $b_3$ and~$b_4$ be the breaklines, numbered from left to right.
    It holds in both output dimensions that:
    \begin{itemize}
        \item $f$ is constantly~$0$ to the left of~$b_1$ and to the right of~$b_4$.
        \item $f$ is constantly~$6$ between~$b_2$ and~$b_3$.
        \item $b_3$ lies on~$\ell_7$ and~$b_4$ lies on~$\ell_{10}$.
        \item The \emph{slope} of the variable gadget, i.e., the norm of the gradient between~$b_1$ and~$b_2$, is at least~$\frac{3}{2}$ and at most~$3$.
    \end{itemize}
\end{lemma}

Before we prove \cref{lem:variable_gadget}, let us describe the functionality of a variable gadget:
The slope~$s_X$ of a variable gadget for a variable~$X$ is in $\bigl[\frac{3}{2},3\bigr]$.
In order to represent values in $\bigl[\frac{1}{2}, 2\bigr]$, we say that a slope~$s_X$ encodes the value~$X = s_X - 1$.

\begin{proof}[Proof of \cref{lem:variable_gadget}]
    We first prove that four breaklines are indeed necessary to fit all data lines exactly.
    Every orthogonal cross-section contains four non-collinear triples of consecutive points: $(p_2, p_3, p_4)$, $(p_4, p_5, p_6)$, $(p_6, p_7, p_8)$ and $(p_9, p_{10}, p_{11})$.
    They pairwise share at most one point, so by \cref{obs:breakpoints}\ref{obs:non_collinear_breakpoints}, four breakpoints are indeed required.
    Since the data lines are parallel to each other, all orthogonal cross-sections look the same, and each breakpoint corresponds to a breakline that is parallel to the data lines.
    For the rest of the proof, we denote the breaklines by~$b_1$, $b_2$, $b_3$ and $b_4$, numbered from left to right.

    In the following, we consider each of the non-collinear triples individually, to further locate the positions of the breaklines.
    The following observations are all due to \cref{obs:breakpoints}:
    \begin{itemize}
        \item The non-collinear triple~$(p_2, p_3, p_4)$ implies that~$b_1$ must be between~$\ell_2$ and~$\ell_4$.
        The collinear triple $(p_1, p_2, p_3)$ enforces that~$b_1$ is to the right of~$\ell_3$ and that~$f$ is constantly~$0$ to the left of~$b_1$.

        \item The non-collinear triple~$(p_4, p_5, p_6)$ implies that~$b_2$ must be between~$p_4$ and~$p_6$.
        The collinear triple~$(p_5, p_6, p_7)$ enforces that~$b_2$ is to the left of~$\ell_5$ and that~$f$ is constantly~$6$ to the right of~$b_2$.

        \item The non-collinear triple~$(p_6, p_7, p_8)$ implies that~$b_3$ must be between~$\ell_6$ and~$\ell_8$.
        The collinear triples~$(p_5, p_6, p_7)$ and~$(p_7, p_8, p_9)$ leave~$\ell_7$ as the only remaining position for~$b_3$.
        The collinear triple $(p_5, p_6, p_7)$ further implies that~$f$ must be constantly~$6$ to the left of~$b_3$.

        \item The non-collinear triple~$(p_9, p_{10}, p_{11})$ implies that~$b_4$ must be between~$\ell_9$ and~$\ell_{11}$.
        The collinear triples~$(p_8, p_9, p_{10})$ and~$(p_{10}, p_{11}, p_{12})$ leave~$\ell_{10}$ as the only remaining position for~$b_4$.
        The collinear triple~$(p_{10}, p_{11}, p_{12})$ further implies that~$f$ must be constantly~$0$ to the right of~$b_4$.
    \end{itemize}

    As~$b_1$ must be on~$\ell_3$ or to its right and~$b_2$ must be on~$p_5$ or to its left, the slope is at least~$\frac{3}{2}$.
    Lastly, weak data point~$q$ enforces that the slope is at most~$3$.
\end{proof}

\subsubsection{Measuring a Value from a Variable Gadget}

Consider a variable gadget for a variable~$X$ with slope~$s_X$.
We call the two parallel lines with distance~$1$ to~$\ell_4$ its \emph{measuring lines}.
More precisely, we distinguish between the \emph{lower measuring line} (the one towards~$\ell_3$) and the \emph{upper measuring line} (the one towards~$\ell_5$).
Since the slope of the variable gadget is in~$\bigl[\frac{3}{2}, 3\bigr]$, both measuring lines are always inside or at the boundary of the sloped part (in other words, between breaklines~$b_1$ and~$b_2$).

Assume that the variable gadget is fit exactly.
Then, at any point~$p$ on~$\ell_4$, the variable gadget contributes~$3$ to~$f(p, \Theta)$.
It follows that a point~$p_u$ on the upper measuring line contributes~$3 + s_X$ to~$f(p_u, \Theta)$.
Similarly, a point~$p_l$ on the lower measuring line contributes~$3 - s_X$ to~$f(p_l,\Theta)$.
See \cref{fig:measuring} for a visualization.

\begin{figure}[htb]
    \centering
    \includegraphics{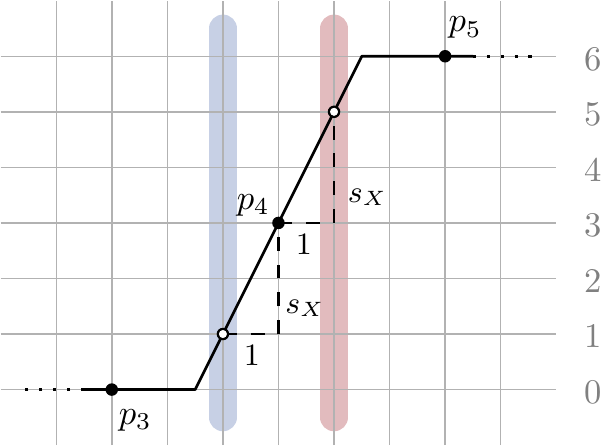}
    \caption{
        Partial cross-section of a variable gadget with slope~$s_X = 2$, so~$X = 1$.
        The lower and upper measuring lines are drawn in blue and red.
        This variable gadget contributes~$3 - s_X$ to the lower and~$3 + s_X$ to the upper measuring line.
    }
    \label{fig:measuring}
\end{figure}

\subsubsection{Linear Constraints: Addition and Copying}
\label{sec:linear_constraints}

Until this point, we considered individual gadgets separately from each other.
As soon as we have two or more gadgets, their corresponding stripes may intersect, leading to interference of the gadgets inside these intersections.
We exploit this to encode linear constraints.

Let~$\calA$ and~$\calB$ be disjoint subsets of the variables.
We can enforce a linear constraint of the form $\sum_{A \in \calA} A = \sum_{B \in \calB} B$ using just one additional data point~$p$.
In particular, we care about the following two special cases:
\begin{itemize}
    \item To copy a value from one variable~$X$ to another variable~$Y$, we model~$X = Y$ by $\calA = \{X\}$ and $\calB = \{Y\}$.
    \item To encode the addition~$X + Y = Z$, we set $\calA = \{X,Y\}$ and $\calB = \{Z\}$.
\end{itemize}

For all variables in~$\calA$, the data point~$p$ must be on the upper measuring line of their corresponding variable gadget.
Similarly, for variables in~$\calB$, the data point~$p$ must be on the lower measuring line.
This requires a placement of the gadgets such that all required measuring lines intersect in a single point, where we can place~$p$.
This is trivial for~$\abs{\calA} + \abs{\calB} = 2$, as it only requires the involved variable gadgets to be non-parallel, see \cref{fig:linear_constraint_2way}.
For~$\abs{\calA} + \abs{\calB} \geq 3$, this is more involved.
We can use the equality constraint~$X = Y$ to copy the value of a variable onto additional variable gadgets, which can then be positioned freely to obtain the required intersections, see \cref{fig:linear_constraint_3way}.
We discuss the global layout to achieve this in more detail in \cref{sec:global_layout} below.

\begin{figure}
    \centering
    \includegraphics[page=1]{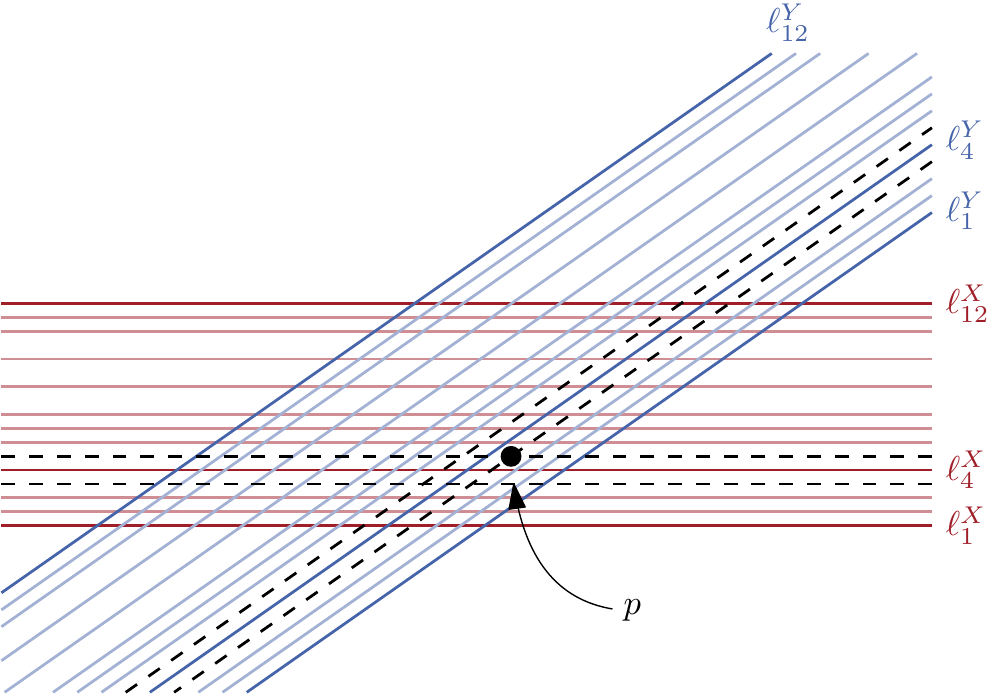}
    \caption{
        Top-down view on the intersection of two variable gadgets corresponding to two variables~$X$~(red) and~$Y$~(blue).
        The dashed lines are their measuring lines.
        The point~$p$ is placed at the intersection of the upper measuring line for~$X$ and lower measuring line for~$Y$, and receives label~$6$ to enforce the constraint~$X = Y$.
    }
    \label{fig:linear_constraint_2way}
\end{figure}

\begin{figure}
    \centering
    \includegraphics[page=2]{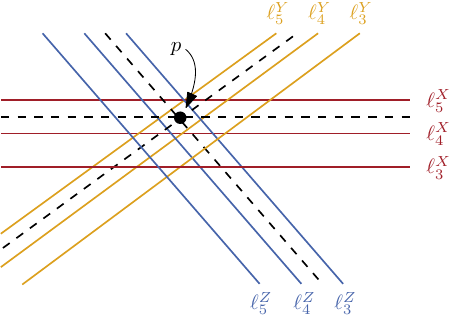}
    \caption{
        Top-down view of the intersection of three variable gadgets corresponding to variables~$X$~(red), $Y$~(orange), and~$Z$~(blue).
        The dashed lines are the upper measuring lines for~$X$ and~$Y$ and the lower measuring line for~$Z$, intersecting in a single point~$p$ with label~$10$.
        This realizes the constraint~$X + Y = Z$.
    }
    \label{fig:linear_constraint_3way}
\end{figure}

\begin{lemma}
    \label{lem:linear_constraints}
    Let~$\calA$ and~$\calB$ be disjoint subsets of the variables.
    A data point~$p$ with label $4\abs{\calA} + 2\abs{\calB}$ placed on the upper measuring line of each~$A \in \calA$ and the lower measuring line of each~$B \in \calB$ enforces the linear constraint $\sum_{A \in \calA} A = \sum_{B \in \calB} B$.
\end{lemma}

\begin{proof}
    Let~$s_X$ be the slope of the variable gadget for variable~$X$.
    For each~$A \in \calA$, the data point~$p$ is placed on the upper measuring line of $A$'s variable gadget, so~$A$ contributes $3 + s_A$ to $f(p, \Theta)$.
    Similarly, for each variable~$B \in \calB$, the data point~$p$ is placed on the lower measuring line of~$B$'s variable gadget, so~$B$ contributes $3 - s_B$ to $f(p, \Theta)$.

    The overall contribution of all involved variables adds up to
    \begin{align*}
        f(p, \Theta)
        &= \sum_{A \in \calA} (3 + s_A) + \sum_{B \in \calB} (3 - s_B) \\
        &= \sum_{A \in \calA} (4 + A) + \sum_{B \in \calB} (2 - B) \\
        &= 4 \abs{\calA} + 2 \abs{\calB} + \sum_{A \in \calA} A - \sum_{B \in \calB} B
        \text{,}
    \end{align*}
    where we used that~$X = s_X - 1$ to get from the first to the second line.
    Setting the label of~$p$ to $4\abs{\calA} + 2\abs{\calB}$ yields that~$p$ is fit exactly if and only if $\sum_{A \in \calA} A = \sum_{B \in \calB} B$.
\end{proof}

As already mentioned above, we do not require \cref{lem:linear_constraints} in its full generality, but only two special cases:
The only linear constraint in an \ETRINV instance is the addition~$X + Y = Z$.
Additionally, we also need the ability to copy values in our reduction, i.e., constraints of the form~$X = Y$.

\begin{corollary}
    \label{cor:linear_constraints_labels}
    To encode the addition constraint $X + Y = Z$ of \ETRINV, data point~$p$ has label~$10$.
    For the copy constraint~$X = Y$, the data point has label~$6$.
\end{corollary}

\subsubsection{Inversion Gadget}

In essence, an \emph{inversion gadget} is the superposition of two variable gadgets.
By using data lines with different labels in the two output dimensions, it can represent two real variables at once.
However, their values have a non-linear dependency.

Formally, the inversion gadget consists of~$13$ parallel data lines~$\ell_1, \ldots, \ell_{13}$, numbered from one side to the other, in the figures from left to right.
The following table lists their relative distance to~$\ell_1$ and their labels:
\begin{center}
    \begin{tabular}{lccccccccccccc}
        \toprule
        & $\ell_1$ & $\ell_2$ & $\ell_3$ & $\ell_4$ & $\ell_5$ & $\ell_6$ & $\ell_7$ & $\ell_8$ & $\ell_9$ & $\ell_{10}$ & $\ell_{11}$ & $\ell_{12}$ & $\ell_{13}$ \\
        \midrule
        distance to $\ell_1$ & $0$ & $1$ & $2$ & $4$ & $7$ & $9$ & $10$ & $11$ & $13$ & $15$ & $17$ & $18$ & $19$ \\
        label in dim. $1$ & $0$ & $0$ & $0$ & $3$ & $6$ & $6$ & $6$ & $6$ & $4$ & $2$ & $0$ & $0$ & $0$ \\
        label in dim. $2$ & $0$ & $0$ & $0$ & $0$ & $3$ & $6$ & $6$ & $6$ & $4$ & $2$ & $0$ & $0$ & $0$ \\
        \bottomrule
    \end{tabular}
\end{center}
See \cref{fig:inversion_gadget} for an orthogonal cross-section through an inversion gadget.

\begin{figure}[htb]
    \centering
    \includegraphics{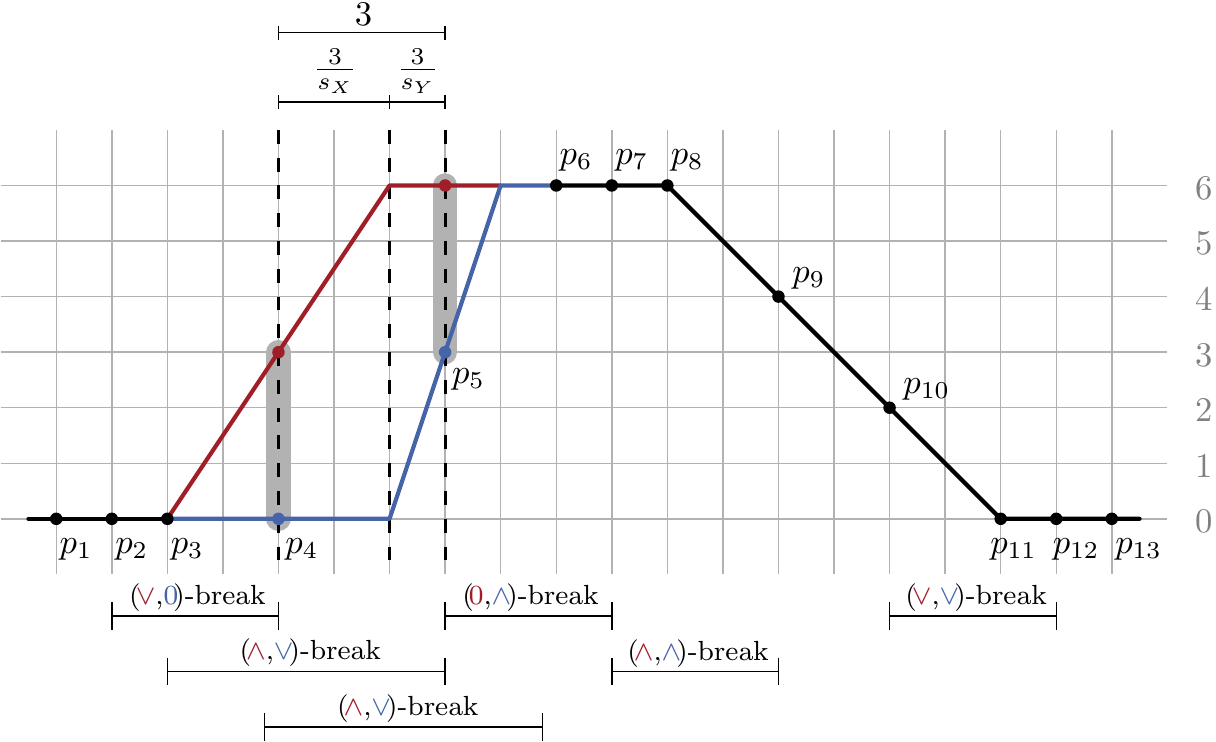}
    \caption{
        Cross-section of the inversion gadget.
        Data points~$p_4$ and~$p_5$ have different labels in the first (red) and second (blue) output dimension.
    }
    \label{fig:inversion_gadget}
\end{figure}

\begin{lemma}
    \label{lem:inversion_gadget}
    A \CPWL function that fits~$\ell_1, \ldots, \ell_{13}$ exactly must have at least five breaklines.
    If it has exactly five breaklines, then they must all be parallel to the data lines.
    In this case, let~$b_1$, $b_2$, $b_3$, $b_4$ and~$b_5$ be the breaklines, numbered from left to right.
    It holds that:
    
    \begin{itemize}
        \item $f$ is constantly~$(0, 0)$ to the left of~$b_1$ and to the right of~$b_5$.
        \item In output dimension~$1$, $f$ is constantly~$6$ between~$b_2$ and~$b_4$.
        \item In output dimension~$2$, $f$ is constantly~$0$ to the left of~$b_2$ and constantly~$6$ between~$b_3$ and~$b_4$.
        \item $b_4$ is on~$\ell_8$ and~$b_5$ is on~$\ell_{11}$.
        \item The inversion gadget has two \emph{slopes}~$s_X$ and~$s_Y$.
        \begin{itemize}[nosep]
            \item In dimension~$1$, slope~$s_X$ is the norm of the gradient between~$b_1$ and~$b_2$.
            \item In dimension~$2$, slope~$s_Y$ is the norm of the gradient between~$b_2$ and~$b_3$.
        \end{itemize}
        It holds that $s_Xs_Y = s_X + s_Y$.
    \end{itemize}
\end{lemma}

Before we prove \cref{lem:inversion_gadget}, let us describe the functionality of an inversion gadget:
As mentioned above, an inversion gadget is the superposition of two variable gadgets.
Only four of the five breaklines are \enquote{visible} in each output dimension, the fifth being erased.
It has a slope in each dimension: $s_X$ in the first, $s_Y$ in the second.
Therefore, it encodes two values $X = s_X - 1$ and~$Y = s_Y - 1$.
It holds that
\[
    XY
    = (s_X - 1)(s_Y - 1)
    = s_Xs_Y - s_X - s_Y + 1
    \stackrel{(*)}{=} s_X + s_Y - s_X - s_Y + 1
    = 1
    \text{,}
\]
where the equality labeled $(*)$ follows from the $s_Xs_Y = s_X + s_Y$ condition provided by \Cref{lem:inversion_gadget}.
We conclude that the non-linear relation between the two slopes exactly models the inversion constraint~$XY = 1$ of an \ETRINV instance.

\begin{proof}[Proof of \cref{lem:inversion_gadget}]
    As in the proof of \cref{lem:variable_gadget}, we use non-collinear triples to argue that each orthogonal cross-section requires exactly five breakpoints.
    Again, because all data lines are parallel to each other, all orthogonal cross-sections look the same, and all breakpoints correspond to a breakline that is parallel to the data lines.

    All the following observations rely on \cref{obs:breakpoints}:
    \begin{itemize}
        \item The non-collinear triple~$(p_2, p_3, p_4)$ enforces a breakline of type~$(\vee, 0)$ strictly between~$\ell_2$ and~$\ell_4$.
        We call this breakline~$b_1$.
        As~$(p_1, p_2, p_3)$ is collinear in both dimensions, $b_1$ must be on or to the right of~$\ell_3$ and~$f$ must be constantly~$(0,0)$ to the left of~$b_1$.

        \item The non-collinear triple~$(p_5, p_6, p_7)$ enforces a breakline of type~$(0, \wedge)$ strictly between~$\ell_5$ and~$\ell_7$.
        We call this breakline~$b_3$.
        As~$(p_6, p_7, p_8)$ is collinear in both dimensions, $b_3$ must be on or to the left of~$\ell_6$ and~$f$ must be constantly~$(6,6)$ to the right of~$b_3$.

        \item The non-collinear triple~$(p_7, p_8, p_9)$ enforces a breakline of type~$(\wedge, \wedge)$ strictly between~$\ell_7$ and~$\ell_9$.
        We call this breakline~$b_4$.
        As $(p_6, p_7, p_8)$ and $(p_8, p_9, p_{10})$ are collinear in both dimensions, $b_4$ must lie on~$\ell_8$.

        \item The non-collinear triple~$(p_{10}, p_{11}, p_{12})$ enforces a breakline of type~$(\vee, \vee)$ strictly between~$\ell_{10}$ and~$\ell_{12}$.
        We call this breakline~$b_5$.
        The triples $(p_9, p_{10}, p_{11})$ and $(p_{11}, p_{12}, p_{13})$ are collinear in both dimensions, so $b_5$ must lie on~$\ell_{11}$.
    \end{itemize}
    The triples enforcing the breaklines~$b_1$, $b_3$, $b_4$ and~$b_5$ are pairwise disjoint, so all of these breaklines are indeed necessary.

    An orthogonal cross-section of an inversion gadget contains two more non-collinear triples:
    $(p_3, p_4, p_5)$ and $(p_4, p_5, p_6)$.
    Both enforce a breakline of type~$(\wedge, \vee)$.
    Note, that all other non-collinear triples intersecting them have a different type.
    Therefore, none of the previous four breaklines is compatible, and at least one more breakline is indeed necessary. 
    Assuming that~$\ell_1, \ldots, \ell_{13}$ must be fit with just five breaklines, we call this breakline~$b_2$, and conclude that it must be on or between~$\ell_4$ and~$\ell_5$.
    Further, $f$ must be constantly~$0$ in dimension~$2$ to the left of~$b_2$.
    
    It remains to prove that~$s_Xs_Y = s_X + s_Y$.
    To this end, we derive the exact position of~$b_2$ between~$\ell_4$ and~$\ell_5$.
    In dimension~$1$, it holds that $f^1(\ell_4, \Theta) = 3$ and $f^1(b_2, \Theta) = 6$.
    The distance between~$\ell_4$ and~$b_2$ is~$\frac{6-3}{s_X}$.
    In dimension~$2$, it holds that~$f^2(\ell_5, \Theta) = 3$ and~$f^2(b_2, \Theta) = 0$.
    The distance between~$\ell_5$ and~$b_2$ is~$\frac{3-0}{s_Y}$.
    We conclude that~$3 = \frac{3}{s_X} + \frac{3}{s_Y}$, which is equivalent to~$s_X + s_Y = s_Xs_Y$ for all~$s_X, s_Y \neq 0$.
\end{proof}

\subsubsection{Applying the Inversion Gadget}

An inversion gadget has two pairs of measuring lines, one in each dimension.
The lower and upper measuring lines in dimension~$1$ have distance~$1$ to~$\ell_4$.
Similar they have distance~$1$ to~$\ell_5$ in dimension~$2$.

To encode an $XY = 1$ constraint of an \ETRINV instance, we first identify two normal variable gadgets carrying the variables~$X$ and~$Y$.
Then the inversion gadget is placed such that its measuring lines intersect the measuring lines of the variable gadgets.
We copy~$X$ to the first dimension of the inversion gadget at the intersection with the variable gadget for~$X$.
Similarly, we copy~$Y$ to the second dimension of the inversion gadget at the intersection with the variable gadget for~$Y$.
This copying can be achieved using weak data points that are only active in the respective dimension (having a \enquote{$\geq\! 0$} label in the other dimension).
See \cref{fig:applying_inversion} for a top-down view.
Technically, the inversion gadget enforces the inversion constraint only in one dimension of each involved variable gadget.
However, this is sufficient because variable gadgets always carry the same value in both output dimensions.

\begin{figure}[htb]
    \centering
    \includegraphics{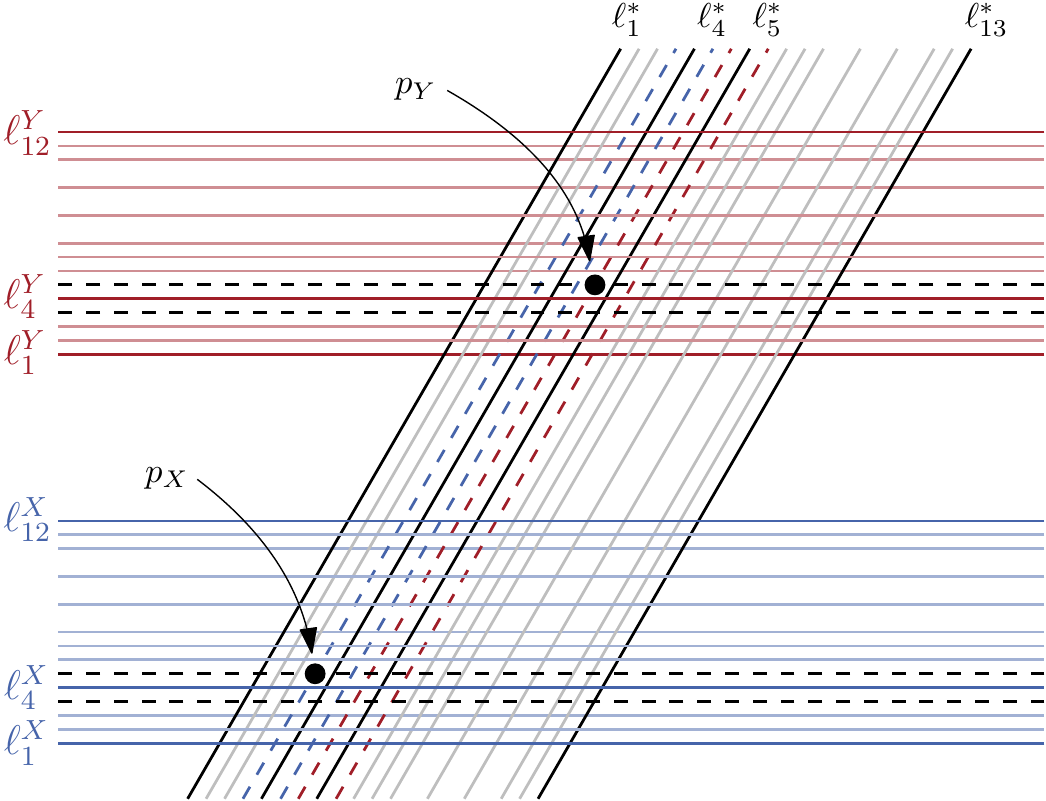}
    \caption{
        Top-down view on two variable gadgets (horizontal) for variables~$X$~(blue) and~$Y$~(red) (the data lines are solid, the measuring lines are dashed).
        The sloped gadget is an inversion gadget.
        Two weak data points~$p_X$ and~$p_Y$ copy~$X$ and~$Y$ to the first and second dimension of the inversion gadget, respectively.
    }
    \label{fig:applying_inversion}
\end{figure}

\subsubsection{Realizing Weak Data Points: Lower Bound Gadgets}
\label{sec:lower_bound_gadget}

So far, we used weak data points, i.e., data points whose label is just a lower bound on~$f(\cdot, \Theta)$.
Weak data points are just a concept meant to simplify the description of the gadgets; a \trainNN instance cannot contain weak data points.
For this reason, we introduce a \emph{lower bound gadget} that simulates a weak data point using only ordinary data points, i.e., data points with constant labels.

Recall that a weak data point can have a lower bound label in either one or in both dimensions.
For this reason, a lower bound gadget can be either \emph{active} or \emph{inactive} in each dimension.
If the lower bound gadget is active in some dimension, its breaklines form a $\vee$-shape of (almost) arbitrary depth in that dimension.
On the other hand, if the lower bound gadget is inactive in some output dimension, then it is constantly~$0$ in this dimension.

Formally, a lower bound gadget consists of eight data lines~$\ell_1, \ldots, \ell_8$, numbered from one side to the other, in the figures from left to right.
The following table lists their relative distance to~$\ell_1$ and their labels:

\begin{center}
    \begin{tabular}{lcccccccc}
        \toprule
        & $\ell_1$ & $\ell_2$ & $\ell_3$ & $\ell_4$ & $\ell_5$ & $\ell_6$ & $\ell_7$ & $\ell_8$ \\
        \midrule
        distance to $\ell_1$ & $0$ & $1$ & $2$ & $3$ & $5$ & $6$ & $7$ & $8$ \\
        label in active dimension(s) & $0$ & $0$ & $0$ & $-1$ & $-1$ & $0$ & $0$ & $0$ \\
        label in inactive dimension & $0$ & $0$ & $0$ & $0$ & $0$ & $0$ & $0$ & $0$ \\
        \bottomrule
    \end{tabular}
\end{center}
See \cref{fig:lower_bound_gadget} for orthogonal cross-sections through a lower bound gadget.
We always assume that a lower bound gadget has at least one active dimension, as it is otherwise unnecessary.

\begin{figure}
    \begin{subfigure}[t]{.47\textwidth}
	\centering
	\includegraphics[page=1]{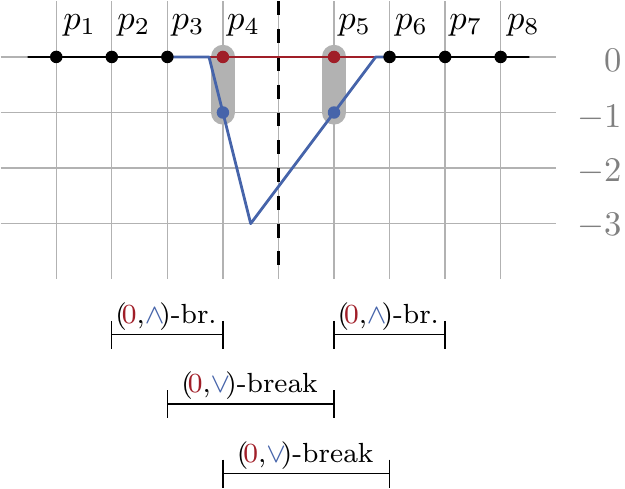}
        \caption{The lower bound gadget can be asymmetric.}
        \label{fig:lower_bound_gadget_asymmetric}
    \end{subfigure}
    \hfill
    \begin{subfigure}[t]{.47\textwidth}
        \centering
        \includegraphics[page=2]{figs/lower_bound_cross_section.pdf}
        \caption{A lower bound gadget has a maximum contribution to the weak data point of~$-2$.}
        \label{fig:lower_bound_gadget_maximum_contribution}
    \end{subfigure}
    \caption{
        Cross-sections of a lower bound gadget which is inactive in the first (red) dimension and active in the second (blue) dimension.
        It is used to simulate a weak data point in the active dimension~(blue) that lies on the dashed vertical line.
        It does not contribute to~$f(\cdot, \Theta)$ in the inactive dimension~(red), i.e., all breaklines are erased (type~$0$).
    }
    \label{fig:lower_bound_gadget}
\end{figure}

\begin{lemma}
    \label{lem:lower_bound_gadget}
    A \CPWL function~$f$ that fits~$\ell_1, \ldots, \ell_8$ exactly must have at least three breaklines.
    If it has exactly three breaklines, then they must all be parallel to the data lines.
    In this case, let~$b_1$, $b_2$ and~$b_3$ be the breaklines.
    In an inactive dimension,~$f$ is constantly~$0$.
    In an active dimension, it holds that:
    \begin{itemize}
        \item $f$ is constantly~$0$ to the left of~$b_1$ and to the right of~$b_3$.
        \item $f(p, \Theta) \leq -2$ for all points~$p$ with equal distance to~$\ell_4$ and~$\ell_5$.
    \end{itemize}
\end{lemma}

\begin{proof}
    As in the proof of \cref{lem:variable_gadget}, we use non-collinear triples to argue that at least three breaklines are necessary for a lower bound gadget with at least one active dimension.
    Again, because all data lines are parallel to each other, all orthogonal cross-sections look the same, and all breakpoints correspond to a breakline that is parallel to the data lines.

    The following observations rely on \cref{obs:breakpoints}.
    In an inactive dimension, all triples are collinear, so $f$ must be constantly~$0$ everywhere.
    From now on, we focus purely on the active dimension(s).
    \begin{itemize}
        \item The non-collinear triple~$(p_2, p_3, p_4)$ enforces a $\wedge$-type breakline strictly between~$\ell_2$ and~$\ell_4$.
        We call this breakline~$b_1$.
        By the collinear triple~$(p_1, p_2, p_3)$, $b_1$ must be on or to the right of~$\ell_3$.
        Further, $f$ must be constantly~$0$ to the left of~$b_1$.

        \item The non-collinear triple~$(p_5, p_6, p_7)$ enforces a $\wedge$-type breakline strictly between~$\ell_5$ and~$\ell_7$.
        We call this breakline~$b_3$.
        By the collinear triple~$(p_6, p_7, p_8)$, $b_3$ must be on or to the left of~$\ell_6$.
        Further, $f$ must be constantly~$0$ to the right of~$b_3$.
    \end{itemize}
    Breaklines~$b_1$ and~$b_3$ are both necessary, as the non-collinear triples enforcing them are disjoint.
    The lower bound gadget contains two more non-collinear triples:
    $(p_3, p_4, p_5)$ and $(p_4, p_5, p_6)$, both of type~$\vee$.
    Assuming that all data lines are fit with just three breaklines, a third breakline called~$b_2$ must lie on or between~$p_4$ and~$p_5$.
    The value of~$f(b_2, \Theta)$ can be arbitrarily small.
    However, in the extreme case, it holds that~$b_1 = \ell_4$ and~$b_3 = \ell_5$, in which case~$f(b_2, \Theta) = -2$.
    No larger values are possible.
\end{proof}

For a lower bound gadget with two active dimensions, we can make the following \lcnamecref{obs:lower_bound_gadget_both_dimensions}.
It holds because the breaklines must be at the exact same positions in both output dimensions.

\begin{observation}
    \label{obs:lower_bound_gadget_both_dimensions}
    A lower bound gadget that is active in both dimensions contributes the same amount in both dimensions.
\end{observation}

We need one lower bound gadget per weak data point.
It gets placed such that the weak data point is equidistant to~$\ell_4$ and~$\ell_5$.
The weak data point with label~$\geq y$ is converted into an ordinary data point with label~$y-2$.

By \cref{lem:lower_bound_gadget}, the lower bound gadget can contribute any value $c \in (-\infty,-2]$ to the new data point.
Thus, the data point can be fit perfectly if and only if the other gadgets contribute at least a value of~$y$ to the data point, that is, the intended lower bound constraint is met.

\subsubsection{Realizing Data Lines using Data Points}
\label{sec:data_lines_to_data_points}

We previously assumed that our gadgets are defined by data \emph{lines}, while in reality, we are only allowed to use data \emph{points}.
In this \lcnamecref{sec:data_lines_to_data_points}, we argue that a set of data lines can be simulated by replacing each data line by three data points.
This allows us to define the gadgets described throughout previous sections solely using data points.

This \lcnamecref{sec:data_lines_to_data_points} is devoted to showing the following \lcnamecref{lem:data_lines_to_data_points}, which captures this transformation formally.
Note that our replacement of data lines by data points does not work in full generality, but we show it for all the gadgets that we constructed.

\begin{lemma}
    \label{lem:data_lines_to_data_points}
    Assume we are given a set of variable, inversion and lower bound gadgets that in total requires at least~$m$ breaklines (four, five, and three per variable, inversion and lower bound gadget, respectively).
    Further, let the gadgets be placed in~$\R^2$ such that no two parallel gadgets overlap.
    Then each data line can be replaced by three data points, such that a \CPWL function with at most~$m$ breaklines fits the data points if and only if it fits the data lines.
\end{lemma}

For the proof, consider the line arrangement induced by the data lines.
We introduce three vertical lines~$v_1, v_2, v_3$ to the right of all intersections.
The vertical lines are placed at unit distance to one another.
In our construction in \cref{sec:global_layout}, we make sure that no data line is vertical.
Thus, each data line intersects each of the vertical lines exactly once.
We place one data point on each intersection between a vertical line and a data line.
The new data point inherits the label of the underlying data line.
Furthermore, on each vertical line, we ensure that the minimum distance~$\alpha$ between any two data points belonging to different gadgets is larger than the maximum distance~$w$ between data points belonging to the same gadget.
This can be achieved by placing the~$v_1$,~$v_2$ and~$v_3$ far enough to the right and by ensuring a minimum distance between parallel gadgets.
See \cref{fig:data_lines_to_data_points} for an illustration.

\begin{figure}[htb]
    \centering
    \includegraphics{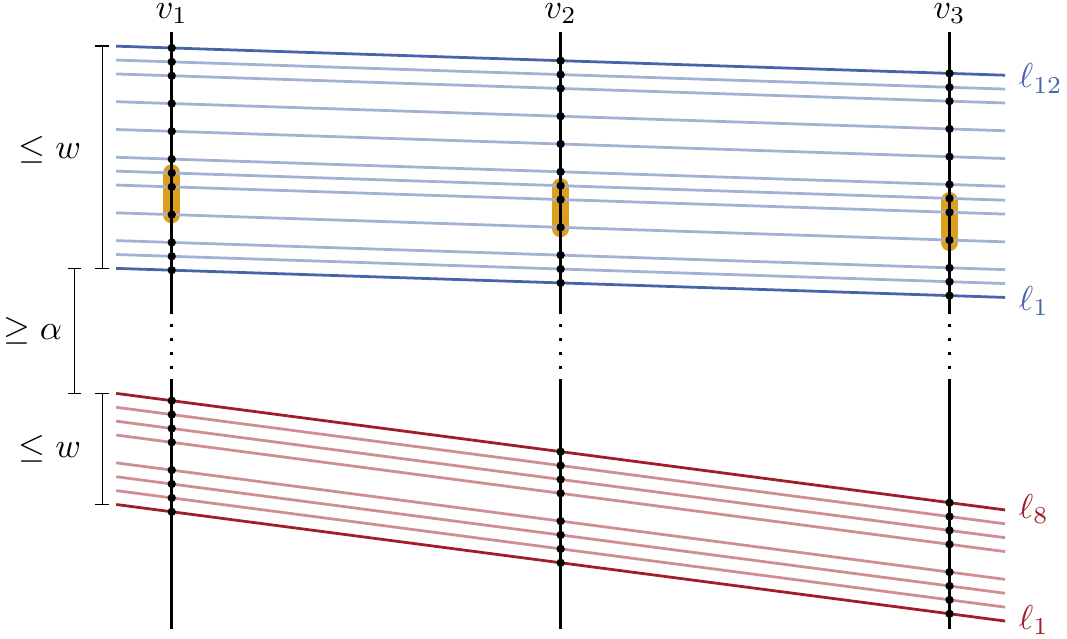}
    \caption{
        Data lines defining a variable gadget~(blue) and a lower bound gadget~(red), and their intersections with the vertical lines~$v_1,v_2,v_3$.
        We add a data point at each intersection.
        The values~$\alpha$ and~$w$ describe the minimal distance between data lines of different gadgets, and the maximal distance between data lines of the same gadget, respectively.
        In orange, we highlighted three matching breakpoint intervals (forcing a $\wedge$-breakpoint between~$\ell_4$ and~$\ell_6$ of the variable gadget).
    }
    \label{fig:data_lines_to_data_points}
\end{figure}

Along each of the three vertical lines, the data points form cross-sections of all the gadgets, similar to the cross-sections shown in \cref{fig:variable_gadget,fig:inversion_gadget,fig:lower_bound_gadget} (but here, the cross-sections are not orthogonal).
We have previously analyzed cross-sections of individual gadgets in the proofs of \cref{lem:variable_gadget,lem:inversion_gadget,lem:lower_bound_gadget}.
There, we identified certain intervals between some of the data lines that need to contain a \emph{breakpoint} (the intersection of a breakline and the cross section).
We refer to these intervals as \emph{breakpoint intervals} along the vertical lines.
Note that a breakpoint interval may degenerate to just one point.
By our placement of the vertical lines, the cross-sections (and thus also the breakpoint intervals) of different gadgets do not overlap.

Any two data lines bounding a breakpoint interval on~$v_1$ also bound a breakpoint interval on~$v_2$ and~$v_3$.
We call the three breakpoint intervals on~$v_1$, $v_2$ and~$v_3$ which are bounded by the same data lines \emph{matching} breakpoint intervals.

In total, there are~$3m$ breakpoint intervals.
We show that the only way to \emph{stab} each of them exactly once using~$m$ breaklines is if each breakline stabs exactly three matching breakpoint intervals.
The first observation towards this is that each breakline can only stab a single breakline interval per vertical line because all breakline intervals are pairwise disjoint.
Thus, having~$m$ breakpoint intervals on each vertical line, each of the~$m$ breaklines has to stab exactly three intervals, one per vertical line.
In a first step, we show that each breakline has to stab three breakpoint intervals belonging to the same gadget.

\begin{claim}
    \label{claim:multiple_gadgets}
    Each breakline has to stab three breakpoint intervals of the same gadget.
\end{claim}

\begin{claimproof}
    The proof is by induction on the number of gadgets.
    For a single gadget, the claim holds trivially.
    For the inductive step, we consider the lowest gadget~$g$ (on~$v_1$, $v_2$ and~$v_3$) and assume for the sake of contradiction that there is a breakline~$b$ stabbing a breakpoint interval of~$g$ on~$v_2$ and a breakpoint interval of a different gadget~$g'$ above~$g$ on~$v_1$.
    By construction, the minimum distance~$\alpha$ between different gadgets is larger than the maximum width~$w$ of any gadget on all three vertical lines.
    Thus, the distance of any breakpoint interval of~$g'$ to any breakpoint interval of~$g$ on~$v_1$ is larger than the width of~$g$ on~$v_3$.
    Therefore, we know that the breakline~$b$ intersects~$v_3$ below any breakpoint intervals of~$g$, which is the lowest gadget on~$v_3$.
    Thus, it stabs at most two breakpoint intervals in total, and therefore not all intervals can be stabbed.
    The same reasoning holds if the roles of~$v_1$ and~$v_3$ are flipped.
    All breaklines stabbing breakpoint intervals of~$g$ on~$v_2$ must therefore also stab breakpoint intervals of~$g$ on~$v_1$ and~$v_3$.
    Applying the induction hypothesis on the remaining gadgets, it follows that each breakline only stabs breakpoint intervals of the same gadget.
\end{claimproof}

We can therefore analyze the situation for each gadget in isolation.
The main idea is to distinguish by the type of the required breakline.
Each breakline must stab three breakpoint intervals of the same type.
Let us summarize the findings about required breakline locations and types from the proofs of \cref{lem:variable_gadget,lem:inversion_gadget,lem:lower_bound_gadget} in \cref{tab:breakpoint_types}.

\begin{table}[htb]
    \centering
    \caption{Location and type of the breaklines in variable gadgets, inversion gadgets, and lower bound gadgets.}
    \begin{subtable}[b]{0.3\textwidth}
        \centering
        \begin{tabular}{llr}
            \toprule
            & Location & Type \\
            \midrule
            $b_1$ & $[\ell_3,\ell_4)$ & $(\vee,\vee)$ \\
            $b_2$ & $(\ell_4,\ell_5]$ & $(\wedge,\wedge)$ \\
            $b_3$ & on $\ell_7$ & $(\wedge,\wedge)$ \\
            $b_4$ & on $\ell_{10}$ & $(\vee,\vee)$ \\
            \bottomrule
        \end{tabular}
        \vspace{4.8mm}
        \caption{Variable gadget.}
        \label{tab:breakpoint_types_variable_gadget}
    \end{subtable}
    \hfill
    \begin{subtable}[b]{0.3\textwidth}
        \centering
        \begin{tabular}{llc}
            \toprule
            & Location & Type \\
            \midrule
            $b_1$ & $[\ell_3,\ell_4)$ & $(\vee,0)$ \\
            $b_2$ & $(\ell_4,\ell_5)$ & $(\wedge,\vee)$ \\
            $b_3$ & $(\ell_5,\ell_6]$ & $(0,\wedge)$ \\
            $b_4$ & on $\ell_8$ & $(\wedge,\wedge)$ \\
            $b_5$ & on $\ell_{11}$ & $(\vee,\vee)$ \\
            \bottomrule
        \end{tabular}
        \caption{Inversion gadget.}
        \label{tab:breakpoint_types_inversion_gadget}
    \end{subtable}
    \hfill
    \begin{subtable}[b]{0.3\textwidth}
        \centering
        \begin{tabular}{llc}
            \toprule
            & Location & Type \\
            \midrule
            $b_1$ & $[\ell_3,\ell_4)$ & $(0,\wedge)$ \\
            $b_2$ & $(\ell_4,\ell_5)$ & $(0,\vee)$ \\
            $b_3$ & $(\ell_5,\ell_6]$ & $(0,\wedge)$ \\
            \bottomrule
        \end{tabular}
        \vspace{9.9mm}
        \caption{Lower bound gadget.}
        \label{tab:breakpoint_types_lower_bound_gadget}
    \end{subtable}
    \label{tab:breakpoint_types}
\end{table}

\begin{claim}
    \label{claim:single_gadget_variable_gadget}
    To stab all breakpoint intervals of a variable gadget with only four breaklines, each of them has to stab three matching breakpoint intervals.
\end{claim}

\begin{claimproof}
    See \cref{tab:breakpoint_types_variable_gadget}.
    On the three vertical lines, there are six breakpoint intervals for breaklines of type $(\vee,\vee)$ in total.
    If only two breaklines should stab these six breakpoint intervals, one breakline needs to stab at least two of the single-point intervals.
    If a breakline goes through two of the single points, it also goes through the third point, and can thus not go through the proper intervals.
    Therefore, one breakline must stab the single-point intervals, and the other one stabs the proper breakpoint intervals.

    The same argument can be made for the breakpoint intervals of type $(\wedge,\wedge)$, and thus each breakline stabs three matching breakpoint intervals.
\end{claimproof}

\begin{claim}
    \label{claim:single_gadget_inversion_gadget}
    To stab all breakpoint intervals of an inversion gadget with only five breaklines, each of them has to stab three matching breakpoint intervals.
\end{claim}

\begin{claimproof}
    See \cref{tab:breakpoint_types_inversion_gadget}.
    All five sets of three matching breakpoint intervals have a different type of required breakline, thus each breakline stabs three matching breakpoint intervals.
\end{claimproof}

\begin{claim}
    \label{claim:single_gadget_lower_bound_gadget}
    To stab all breakpoint intervals of a lower bound gadget with only three breaklines, each of them has to stab three matching breakpoint intervals.
\end{claim}

\begin{claimproof}
    See \cref{tab:breakpoint_types_lower_bound_gadget}.
    There is only one set of three breakpoint intervals for a breakline of type $(0,\vee)$, so it is trivially matched correctly.

    We can see that the breakpoint intervals for a breakline of type~$(0,\wedge)$ have distance~$2$ from each other, each having width~$1$.
    If the two breaklines of this type would not stab three matching breakpoint intervals, one of them would need to stab two matching intervals and one non-matching interval.
    As the distance between the vertical lines is equal, and the breakpoint intervals are further apart from each other than their width, there is no way for a breakline to lie in this way.
    We conclude that all breaklines stab three matching breakpoint intervals.
\end{claimproof}

It also follows from \cref{claim:single_gadget_variable_gadget,claim:single_gadget_inversion_gadget,claim:single_gadget_lower_bound_gadget} that no two breaklines can cross each other between the vertical lines.
Neither can a breakline cross a data line.
Together with \cref{claim:multiple_gadgets}, we can finally prove \cref{lem:data_lines_to_data_points}.

\begin{proof}[Proof of \cref{lem:data_lines_to_data_points}]
    By \cref{claim:multiple_gadgets}, every breakline must stab three breakpoint intervals of the same gadget.
    By \cref{claim:single_gadget_variable_gadget,claim:single_gadget_inversion_gadget,claim:single_gadget_lower_bound_gadget}, each breakline must stab three matching breakpoint intervals, and therefore the breaklines do not cross any data lines between the three vertical lines.

    It remains to show that the data points already ensure that each breakline~$b$ is parallel to the two parallel data lines~$d$ and~$d'$ enclosing it.
    To this end, consider the parallelogram defined by $d, d', v_1, v_3$ (see \cref{fig:breakline_parallel}) and let~$j$ be an output dimension in which~$b$ is not erased.
    Since no other breakline intersects this parallelogram, we obtain that~$f^j$ has exactly two linear pieces within the parallelogram, which are separated by~$b$.
    Moreover, since~$b$ stabs matching breakpoint intervals, the three data points on~$d$ must belong to one of the pieces.
    Since these points have the same label, it follows that the gradient of this piece in output dimension~$j$ must be orthogonal to~$d$ (and, thus, to~$d'$ as well).
    Applying the same argument on the data points on~$d'$, we obtain that the gradient of the other piece must be orthogonal to~$d$ and~$d'$ as well.
    This implies that also the difference of the gradients of the two pieces is orthogonal to~$d$ and~$d'$.
    Finally, since~$b$ must be orthogonal to this difference of gradients, we obtain that it is parallel to~$d$ and~$d'$.
\end{proof}

\begin{figure}[htb]
    \centering
    \includegraphics{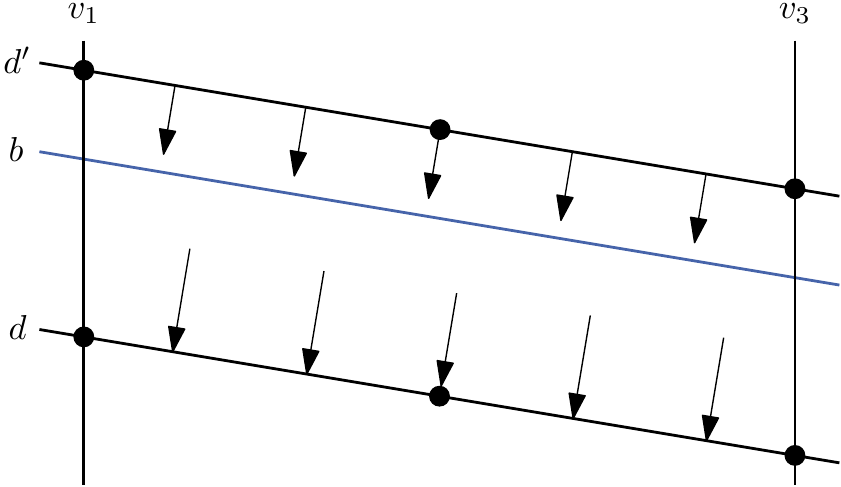}
    \caption{
        The parallelogram enclosed by the two data lines~$d$,~$d'$ and the vertical lines~$v_1$,~$v_3$.
        The three data points (black) on each data line enforce the gradient in both cells to be orthogonal to the data lines.
        As a consequence, the breakline~$b$~(blue) separating the cells has to be parallel to the data lines.
    }
    \label{fig:breakline_parallel}
\end{figure}

\subsection{Global Construction}
\label{sec:global_layout}

As a last step in our \ER-hardness proof of \trainNN, we describe the global arrangement of the different gadgets.
To this end, fix an arbitrary \ETRINV instance.
See \cref{fig:global_arrangement} for a visualization.

\begin{figure}[tb]
    \centering
    \includegraphics[page=2]{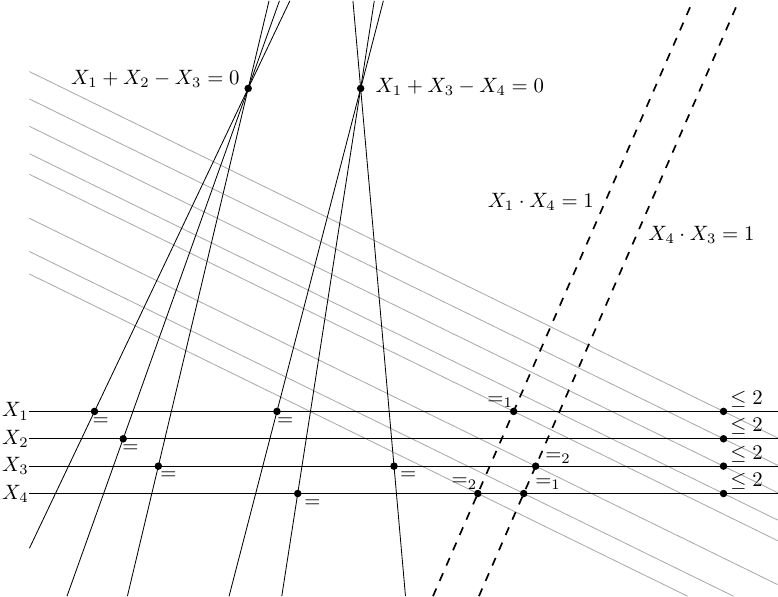}
    \caption{
        The layout of all gadgets and additional data points for the complete reduction.
        Each gadget is simplified to a single line for clarity.
        Solid:~Variable gadgets.
        Dashed:~Inversion gadgets.
        Gray:~Lower bound Gadgets.
        A point with label~$=_i$ indicates a copy that is only active in output dimension~$i$.
    }
    \label{fig:global_arrangement}
\end{figure}

\begin{description}
    \item[Variables] For each variable~$X$, we build a horizontal variable gadget carrying the value of this variable.
    We say that this is the \emph{canonical} variable gadget for~$X$.
    \Cref{lem:variable_gadget} already ensures that $X \in \bigl[\frac{1}{2}, 2\bigr]$.

    In order to realize the weak data points of the variable gadgets, we add one lower bound gadget each.
    They are placed parallel to each other, but not parallel to the variable gadgets.
    Further, their stripes must not intersect any other data points.

    \item[Addition] The following is done for each addition constraint~$X + Y = Z$, next to each other:
    For each involved variable, we copy the value from its canonical variable gadget to a new variable gadget.
    To this end, a data point with label~$6$ is placed on the intersection of the upper measuring line of the canonical variable gadget and the lower measuring line of the new variable gadget (\cref{cor:linear_constraints_labels}).

    Further, the three new variable gadgets are positioned such that the correct measuring lines intersect in a common intersection above all horizontal variable gadgets (upper for~$X,Y$ and lower for~$Z$).
    A data point with label~$10$ at the intersection enforces the addition constraint (\cref{cor:linear_constraints_labels}).

    \item[Inversion] The following is done for each inversion constraint~$XY = 1$, next to each other:
    We add an inversion gadget that intersects all canonical variable gadgets.
    Using weak data points with label~$6$ in the respective dimensions, we copy $X$ to the first dimension of the inversion gadget, and~$Y$ to its second dimension (\cref{cor:linear_constraints_labels}).
    The inversion constraint itself is then enforced by the non-linear relation of the two slopes of the inversion gadget (\cref{lem:inversion_gadget}).

    All inversion gadgets are placed parallel to each other.

    The involved weak data points are realized by two lower bound gadgets per inversion gadget.
    We place them parallel to the lower bound gadgets for the variables, again such that their stripes do not intersect any other data points.
\end{description}

Based on the global arrangement of the gadgets, we can finally prove our main \lcnamecref{thm:er_complete}, i.e., the \ER-completeness of \trainNN:

\begin{proof}[Proof of \cref{thm:er_complete}]
    For \ER-membership we refer to~\cite{Abrahamsen2021_NeuralNetworks} and to \cref{sec:membership}.
    
    For \ER-hardness, we reduce the \ER-complete problem \ETRINV to \trainNN.
    Given an instance of \ETRINV, we construct an instance of \trainNN as described in the previous paragraphs.
    We set the target error to~$\gamma = 0$.

    Let~$m$ be the minimum number of breaklines needed to realize all gadgets of the above construction:
    We need four breaklines per variable gadget (\cref{lem:variable_gadget}), five breaklines per inversion gadget (\cref{lem:inversion_gadget}) and three breaklines per lower bound gadget (\cref{lem:lower_bound_gadget}).
    By \cref{lem:data_lines_to_data_points}, no breakline can contribute to different gadgets, so we need exactly that many.

    In the remainder of the proof, we show equivalence of the following statements.
    \begin{enumerate}[label=(\arabic*)]
        \item\label{itm:ETRINV_yes}
        The \ETRINV instance is a yes-instance, i.e., there exists a satisfying assignment of the variables.

        \item\label{itm:fittable}
        There exists a \CPWL function with~$m$ breaklines that fits all data points of the constructed \trainNN instance.
        Further, it fulfills the conditions of \cref{lem:fittableCPWL}.

        \item\label{itm:trainNN_yes}
        The \trainNN instance is a yes-instance, i.e., there exists a \fullyconnected \twolayer neural network with~$m$ hidden \ReLU neurons exactly fitting all the data points.
    \end{enumerate}

    To see that~\ref{itm:ETRINV_yes} implies~\ref{itm:fittable}, assume that there is a satisfying assignment of the \ETRINV instance with all variables in~$\bigl[\frac{1}{2}, 2\bigr]$.
    For each variable~$X$, we use~$s_X = X + 1$ as the slope of all corresponding variable gadgets and inversion gadgets.
    The superposition of all these gadgets yields the desired \CPWL function.
    It satisfies \cref{lem:fittableCPWL} because, first, the gadgets are built in such a way that functions fitting all data points are constantly zero everywhere except for within the gadgets, and second, the gradient condition is satisfied for each gadget separately and, hence, also for the whole function.

    For the other direction, i.e., that \ref{itm:fittable} implies~\ref{itm:ETRINV_yes}, assume that such a \CPWL function exists.
    By \cref{lem:data_lines_to_data_points,lem:lower_bound_gadget}, the data points enforce exactly the same \CPWL function as the conceptual data lines and weak data points would.
    Then, by \cref{lem:variable_gadget,lem:inversion_gadget}, this \CPWL function has the shape of the gadgets.
    The fact that all data points are fit implies that the slopes of the variable and inversion gadgets indeed correspond to a satisfying assignment of \ETRINV.

    \Cref{lem:fittableCPWL} yields that~\ref{itm:fittable} implies~\ref{itm:trainNN_yes}.

    It remains to prove that~\ref{itm:trainNN_yes} implies~\ref{itm:fittable}.
    To this end, first note that the function realized by a \fullyconnected \twolayer neural network with~$m$ hidden \ReLU neurons is always a \CPWL function with at most~$m$ breaklines that additionally satisfies the gradient condition by \cref{lem:fittableCPWL}.
    The existence of a~$(0,0)$-cell follows from the fact, that all gadgets are constantly~$0$ outside their stripes.

    \medskip
    \noindent
    The \trainNN instance can be constructed in polynomial time, as the gadgets can be arranged in such a way that all data points (residing on intersections of lines) have coordinates which can be encoded in polynomial length.

    The number of hidden neurons~$m$ is linear in the number of variables and the number of constraints of the \ETRINV instance.
    The number of data points can be bounded by~$10m$, thus the number of hidden neurons is linear in the number of data points.

    As can be gathered from \cref{lem:variable_gadget,lem:inversion_gadget,lem:lower_bound_gadget,cor:linear_constraints_labels}, the set of used labels has cardinality~$13$ as claimed.
\end{proof}

\begin{remark}
    \label{rem:lipschitz}
    Note that if the \ETRINV instance is satisfiable, then each variable gadget and inversion gadget in a corresponding solution~$\Theta$ to the constructed \trainNN instance has a maximum slope of~$3$ in each dimension.
    Furthermore, no lower bound gadget needs to contribute less than~$-12$ to satisfy its corresponding weak data point.
    Thus, there must also be a solution~$\Theta'$, where each lower bound gadget is symmetric, and thus the function~$f(\cdot,\Theta')$ is Lipschitz continuous with a low Lipschitz constant~$L$, which in particular does not depend on the given \ETRINV instance.
    Checking all the different ways how our gadgets intersect, one can verify that~$L = 25$ is sufficient.
\end{remark}

\section{Algebraic Universality}
\label{sec:universality}

It remains to prove algebraic universality of \trainNN.
Intuitively, it suffices to show that a solution of an \ETRINV instance can be transformed into a solution of the corresponding \trainNN instance (and vice versa) using only basic field arithmetic, that is, addition, subtraction, multiplication, and division.

For the following \lcnamecref{lem:algebraic_translation}, let~$\Phi$ be an instance of~$\ETRINV$ with~$k$ variables, and let~$N$ be an instance of \trainNN with a total of~$\ell$ weights and biases.
Further, $N$ was constructed from~$\Phi$ via our reduction.
We denote by~$V(\Phi) \subseteq \R^k$ the set of all satisfying variable assignments.
Similarly,~$V(N) \subseteq \R^\ell$ contains all weight-bias-combinations that fit all data points.

\begin{lemma}
    \label{lem:algebraic_translation}
    For any field extension~$\F$ of~$\Q$ it holds that
    \[
        V(\Phi) \cap \F^k \neq \emptyset
        \iff
        V(N) \cap \F^\ell \neq \emptyset
        \text{.}
    \]
\end{lemma}

\begin{proof}
    \begin{description}
        \item[\enquote{$\Longrightarrow:$}]
        Let $X_1, \ldots, X_n \in V(\Phi) \cap \F^k$ be a satisfying variable assignment for~$\Phi$.
        In our reduction, we place our data points on rational coordinates, and thus all implied data lines can be described by equations with rational coefficients.
        There exists a unique \CPWL function~$f$ which fits these data points, corresponds to the solution~$X_1, \ldots, X_n$, and in which all lower bound gadgets are symmetric.
        This function can be realized by a \fullyconnected \twolayer neural network by \cref{lem:fittableCPWL}.
        The gradients of all cells in this function can be obtained through elementary operations from the values~$X_1, \ldots, X_n$ and rational numbers.
        Furthermore, all breaklines can be described by equations with coefficients derivable from these same numbers.
        Thus, there exist weights and biases $\Theta \in \F^\ell$ for the neural network which realize function~$f$, showing that $\Theta \in V(N) \cap \F^\ell \neq \emptyset$.
        
        \item[\enquote{$\Longleftarrow:$}]
        Let $\Theta \in V(N) \cap \F^\ell$ be a set of weights and biases fitting all data points of the \trainNN instance~$N$.
        For each variable~$X$ of~$\Phi$, there is a canonical variable gadget corresponding to~$X$ whose slope~$s_X$ satisfies~$X = s_X - 1$.
        There is a unique hidden neuron~$v_i$ contributing the first breakline of that variable gadget.
        Using the notation from \cref{sec:geometry_of_NN}, the slope of this variable gadget is $a_{2,i} \cdot c_{i,1}$, because the variable gadget is horizontal (implying that $a_{1,i} = 0$) and its output is equal in both output dimensions (implying $c_{i,1} = c_{i,2}$).
        Thus, $X = a_{2,i} \cdot c_{i,1} - 1$, which is clearly in~$\F$.
        The same holds for all other variables, thus $V(\Phi) \cap \F^k \neq \emptyset$.
        \qedhere
    \end{description}
\end{proof}

Now \cref{thm:universality}, i.e., the algebraic universality of \trainNN, follows directly from the algebraic universality of \ETRINV (\cref{thm:ETRINV_universality}) combined with \cref{lem:algebraic_translation}.

\printbibliography

\end{document}